%% file: certification_coloring.tex
\documentclass[a4paper, USenglish, thm-restate]{lipics-v2021}

\pdfoutput=1

\hideLIPIcs

\usepackage[utf8]{inputenc}
\usepackage[T1]{fontenc}

\usepackage{comment}
\usepackage{graphics}
\usepackage{graphicx}
\usepackage{enumerate}
\usepackage{enumitem}

\usepackage[linesnumbered,ruled,onelanguage,vlined]{algorithm2e}

\usepackage{tikz}
\pgfdeclarelayer{bg}
\pgfsetlayers{bg,main}

\newcommand{\NN}{\mathbb{N}} 

\usetikzlibrary{calc}
\let\le\leqslant
\let\ge\geqslant

\theoremstyle{plain}
\newtheorem{open}[theorem]{Open problem}

\nolinenumbers

\title{Local certification of local properties: tight bounds, trade-offs and new parameters}

\author{Nicolas Bousquet}{Univ Lyon, CNRS, INSA Lyon, UCBL, LIRIS, UMR5205 F-69622 Villeurbanne, France}{nicolas.bousquet@cnrs.fr}{https://orcid.org/0000-0003-0170-0503}{}

\author{Laurent Feuilloley}{Univ Lyon, CNRS, INSA Lyon, UCBL, LIRIS, UMR5205, F-69622 Villeurbanne, France}{laurent.feuilloley@cnrs.fr}{https://orcid.org/0000-0002-3994-0898}{}

\author{Sébastien Zeitoun}{Univ Lyon, CNRS, INSA Lyon, UCBL, LIRIS, UMR5205, F-69622 Villeurbanne, France}{sebastien.zeitoun@ens-lyon.fr}{https://orcid.org/0009-0003-2675-8581}{}

\authorrunning{N. Bousquet, L. Feuilloley and S. Zeitoun}



\acknowledgements{The authors thank anonymous reviewers for useful comments.}

\funding{This work is supported by the ANR grant GrR (ANR-18-CE40-0032). }

\begin{document}
	
	\maketitle
	\begin{abstract}
		Local certification is a distributed mechanism enabling the nodes of a network to check the correctness of the current configuration, thanks to small pieces of information called certificates. 
		For many classic global properties, like checking the acyclicity of the network, the optimal size of the certificates depends on the size of the network, $n$. 
		In this paper, we focus on properties for which the size of the certificates does not depend on $n$ but on other parameters.
		
		We focus on three such important properties and prove tight bounds for all of them. 
		Namely, we prove that the optimal certification size is: $\Theta(\log k)$  for $k$-colorability (and even exactly $\lceil \log k \rceil$ bits in the anonymous model while previous works had only proved a $2$-bit lower bound); $(1/2)\log t+o(\log t)$  for dominating sets at distance $t$ (an unexpected and tighter-than-usual bound) ; and $\Theta(\log \Delta)$ for perfect matching in graphs of maximum degree~$\Delta$ (the first non-trivial bound parameterized by~$\Delta$). 
		We also prove some surprising upper bounds, for example, certifying the existence of a perfect matching in a planar graph can be done with only two bits.   
		In addition, we explore various specific cases for these properties, in particular improving our understanding of the trade-off between locality of the verification and certificate size. 
	\end{abstract}
	
	\newpage{}
	
	\setcounter{page}{1}
	
	\input{intro.tex}

	
	\section{Model and definitions}
	\label{sec:model}
	
	
	\subsection{Graphs}
	
	All the graphs we consider are finite, simple, and non-oriented. 
	For completeness, let us recall the following classical graph definitions. Let $G=(V,E)$ be a graph, $u,v \in V$, $S\subseteq V$, $i \in \NN$.
	The \emph{distance between $u$ and~$v$}, denoted by $d(u,v)$, is the length (number of edges) of the shortest path from $u$ to~$v$.
	The \emph{layer at distance $i$ from~$u$}, denoted by $N^i(u)$, is the set of vertices $v\in V$ such that $d(u,v)=i$.
	The \emph{ball of radius $i$ centered in~$u$}, is $B(u,i):=\bigcup_{0 \leqslant j \leqslant i}N^j(u)$.
	The \emph{closed (resp. open) neighborhood} of $u$ is $N[u]:=B(u,1)$ (resp. $N^1(u)$).
	We can similarly define $d(u,S), N(S)$ and $N[S]$ when $S$ is a subset of vertices. 
	

	\subsection{Certification}
	\label{subsec:model-certification}
	
	Let $G=(V,E)$ be a graph, and let $C,I$ be non-empty sets. A \emph{certificate function of $G$} (with certificates in $C$) is a mapping $c : V \rightarrow C$. An \emph{identifier assignment of $G$} (with identifiers in~$I$) is an injective mapping $Id : V \rightarrow I$.
	
	\begin{definition}
		Let $c$ be a certificate function of $G$ and $Id$ be an identifier assignment of~$G$. Let $u \in V$, $d \in \NN^\ast$. The \emph{view of $u$ at distance $d$} consists in all the information available at distance at most $d$ from $u$, that is:
		\begin{itemize}
			\item the vertex $u$;
			\item the graph with vertex set $B(u,d)$ and the edges $(v_1,v_2)\in E(G)$ such that\\ \mbox{$\{v_1,v_2\}\cap B(u,d-1)\neq\emptyset$};
			\item the restriction of $c$ to $B(u,d)$;
			\item the restriction of $Id$ to $B(u,d)$.
		\end{itemize}
	\end{definition}
	
	\begin{remark}
		For a vertex $u$, the subgraph induced by $B(u,d-1)$ is included in the view of $u$ at distance $d$. However, the subgraph induced by $B(u,d)$ is \emph{not} included in the view of $u$ at distance $d$ in general (because the view of $u$ does not contain the edges between two vertices $v_1$ and $v_2$ which are both at distance exactly $d$ from $u$).
	\end{remark}
	
	A \emph{verification algorithm (at distance $d$) in the locally checkable proof model} is a function which takes as input the view (at distance $d$) of a vertex, and outputs a decision, \emph{accept} or \emph{reject}.
	For a property $\mathcal{P}$ on graphs, we say that there is a \emph{certification for $\mathcal{P}$} using $k$ certificates (resp. $k$ bits) if $C$ has size $k$ (resp. $2^k$), and if there exists a verification algorithm~$A$ such that for every graph $G$ and every identifier assignment~$Id$, $G$ has property $\mathcal{P}$ if and only if there exists a certificate function $c$ such that $A$ accepts for every $v\in V(G)$.
	
	A \emph{verification algorithm in the anonymous model} is defined in the exact same way that a verification algorithm in the locally checkable proof model, but vertices are not equipped with a unique identifier (or equivalently, the output is invariant with respect to the identifier assignment).
	
	In this paper, we do not use the proof-labeling scheme model, where the node has access only its own ID, but it appears in a relevant previous work~\cite{ArdevolCFNPR22}.
	
	\paragraph*{Discussion of the anonymous model}
	
	Note that the model we chose as our main model is the anonymous one. Indeed, all our results are for this model, except for the lower bound on colorability which we wanted to strengthen to the model with identifiers in order to fully solve Open Problem 1. 
	
	We think that the natural model for local properties is anonymous. Indeed, the main reason why identifiers are common in local certification is that they are often necessary, which is not the case for local properties. For example, certifying tree-like structures requires certifying that there is a unique connected component, and for this identifiers are needed. Note that in the LOCAL model, identifiers are also used to break symmetry, but in certification, the certificates can do this.
	
	Moreover, the anonymous model is very common in the self-stabilizing literature (see \emph{e.g.}
	\cite{FischerJ06,BernardDPT09,CohenLMPS16}) which is the origin of local certification.
	
	Third, in the paper, we draw a parallel between certification of local properties and LCLs, and the identifiers do not appear in the definition of LCLs.
	
	Finally, in the cases known in the literature where anonymous and ``with-ID'' bounds match (\emph{e.g.} acyclicity), the proofs are similar in spirit, except that the ID case involves more counting arguments (usually assuming that the ID interval is not of linear size), which implies loosing constants everywhere, getting results that are less crisp, and proofs that are more obfuscated.

	\section{Colorability certification}
	
	In this section, we will consider the certification of the \emph{$k$-colorability} property.
	Let us recall that a graph $G$ is said to be $k$-colorable if there exists \emph{proper $k$-coloring} of $G$, that is, is a mapping \mbox{$\varphi:V \rightarrow \{1, \ldots k\}$} such that for all $(u,v)\in E$, $\varphi(u)\neq\varphi(v)$.
	
	For completeness, let us start by proving the following simple upper bound.
	
	\begin{proposition}
		In the anonymous model where vertices can see at distance~1, $k$-colorability can be certified with $\lceil \log k \rceil$ bits.
	\end{proposition}
	
	\begin{proof}
		For a graph $G=(V,E)$ which is $k$-colorable, the certificate function given by the prover is the following. The prover chooses a proper $k$-coloring $\varphi$ of $G$, and assigns certificate $c(u):=\varphi(u)$ to every $u\in V$. The verification algorithm of every vertex $u$ consists in checking if for every neighbor $v$, $c(u)\neq c(v)$. If it is the case, $u$ accepts. Otherwise, $u$~rejects. It is clear that the graph is accepted if and only if it is $k$-colorable.
	\end{proof}
	
	\subsection{Lower bounds}
	
	\ThmColAnon*
	
	\begin{proof}
		Assume by contradiction that there exists a certification of $k$-colorability in the anonymous model using only $k-1$ different certificates. The idea of the proof is to consider a specific $k$-colorable graph $G_k$, for which there must exist an accepting certification function. 
		We will prove that we can flip (\emph{i.e.} cross) two edges of $G_k$ such that the resulting graph $G_k'$ is not $k$-colorable and no vertex is able to detect it (meaning that the local view of each vertex will be unchanged). Thus, in the new graph $G_k'$, each vertex accepts, which is a contradiction since $G_k'$ is not $k$-colorable.
		
		Let us denote by $G_k$ the complete $k$-partite graph, where each set has size $\max(k,3)$. More formally, let $V_1, \ldots, V_k$ be $k$ disjoint sets, each of size $\max(k,3)$. Let $G_k$ be the graph with vertex set \mbox{$V = \bigcup_{i=1}^k V_i$} where $(u,v)$ is an edge if and only if $u$ and $v$ do not belong to the same set $V_i$.
		Since $\{V_1, \ldots, V_k\}$ is a partition of $V$ into $k$ independent sets, $G_k$ is $k$-colorable. Thus, by hypothesis, there exists a certificate function $c : V \rightarrow \{1, \ldots, k-1\}$ such that every vertex accepts.
		By the pigeonhole principle, for every $i \in \{1, \ldots, k\}$ there exist two different vertices \mbox{$u_i, v_i \in V_i$} such that \mbox{$c(u_i)=c(v_i)$} (since each set has size at least~$k$). 
		Again, by the pigeonhole principle, there exist $i\neq j$ such that \mbox{$c(u_i)=c(u_j)$}. Thus, we get $c(u_i)=c(v_i)=c(u_j)=c(v_j)$,  with $u_i, v_i\in V_i$ and $u_j,v_j\in V_j$. Let $G'_k$ be the graph obtained from $G_k$ by removing the two edges $(u_i,u_j), (v_i,v_j)$ and adding the two edges $(u_i,v_i),(u_j,v_j)$, as depicted on Figure~\ref{fig:G, G'}.

		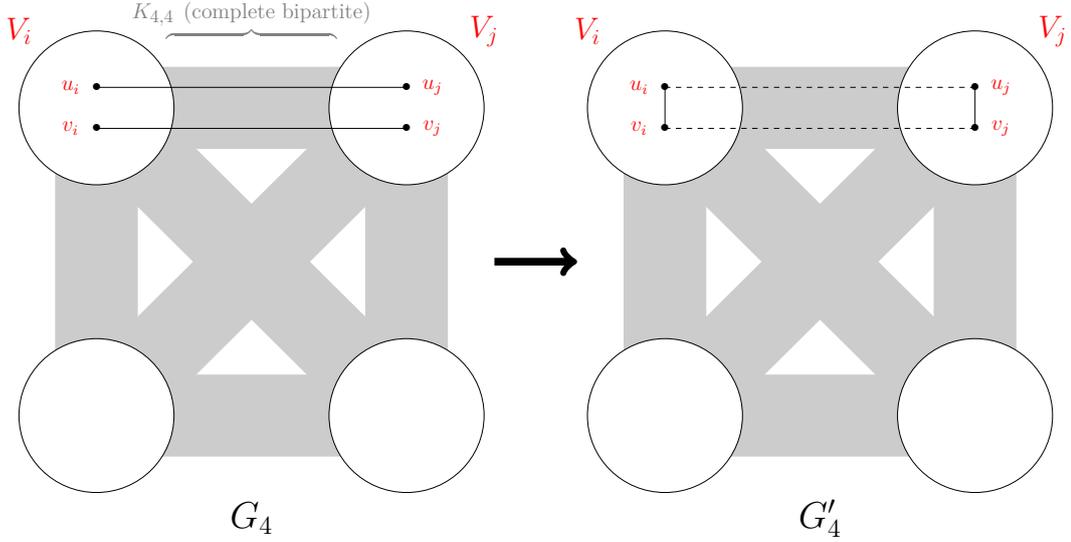
\begin{figure}[h!]
			\centering
			\scalebox{0.68}{
				\begin{tikzpicture}
					\node[circle, draw, minimum size=3cm, fill=white] at (0, 0) () {};
					\node[circle, draw, minimum size=3cm, fill=white] at (6, 0) () {};
					\node[circle, draw, minimum size=3cm, fill=white] at (0, -6) () {};
					\node[circle, draw, minimum size=3cm, fill=white] at (6, -6) () {};
					\node[] at (0,0.4) () {$\bullet$};
					\node[] at (0,-0.4) () {$\bullet$};
					\node[] at (6,0.4) () {$\bullet$};
					\node[] at (6,-0.4) () {$\bullet$};
					
					\draw
					(0,0.4) -- (6,0.4)
					(0,-0.4) -- (6,-0.4);
					
					\begin{pgfonlayer}{bg}
						\draw
						(0,0) edge[line width=16mm, gray!40] node{} (6,0)
						(0,0) edge[line width=16mm, gray!40] node{} (0,-6)
						(0,0) edge[line width=16mm, gray!40] node{} (6,-6)
						(6,0) edge[line width=16mm, gray!40] node{} (6,-6)
						(0,-6) edge[line width=16mm, gray!40] node{} (6,-6)
						(6,0) edge[line width=16mm, gray!40] node{} (0,-6);
					\end{pgfonlayer}
					
					\node[gray] at (3, 1.7) () {\LARGE$\overbrace{~~~~~~~~~~~~~~~~~~}^{K_{4,4}~{\large\text{(complete bipartite)}}}$};
					\node[red] at (-0.5, 0.4) () {\Large$u_i$};
					\node[red] at (-0.5, -0.4) () {\Large$v_i$};
					\node[red] at (6.5, 0.4) () {\Large$u_j$};
					\node[red] at (6.5, -0.4) () {\Large$v_j$};
					\node[red] at (-1.5, 1.5) () {\LARGE$V_i$};
					\node[red] at (7.5, 1.5) () {\LARGE$V_j$};

					\node[circle, draw, minimum size=3cm, fill=white] at (11, 0) () {};
					\node[circle, draw, minimum size=3cm, fill=white] at (17, 0) () {};
					\node[circle, draw, minimum size=3cm, fill=white] at (11, -6) () {};
					\node[circle, draw, minimum size=3cm, fill=white] at (17, -6) () {};
					\node[] at (11,0.4) () {$\bullet$};
					\node[] at (11,-0.4) () {$\bullet$};
					\node[] at (17,0.4) () {$\bullet$};
					\node[] at (17,-0.4) () {$\bullet$};
					
					\draw[dashed]
					(11,0.4) -- (17,0.4)
					(11,-0.4) -- (17,-0.4);
					\draw
					(11,0.4) -- (11,-0.4)
					(17,0.4) -- (17,-0.4);
					
					\begin{pgfonlayer}{bg}
						\draw
						(11,0) edge[line width=16mm, gray!40] node{} (17,0)
						(11,0) edge[line width=16mm, gray!40] node{} (11,-6)
						(11,0) edge[line width=16mm, gray!40] node{} (17,-6)
						(17,0) edge[line width=16mm, gray!40] node{} (17,-6)
						(11,-6) edge[line width=16mm, gray!40] node{} (17,-6)
						(17,0) edge[line width=16mm, gray!40] node{} (11,-6);
					\end{pgfonlayer}
					
					\node[red] at (10.5, 0.4) () {\Large$u_i$};
					\node[red] at (10.5, -0.4) () {\Large$v_i$};
					\node[red] at (17.5, 0.4) () {\Large$u_j$};
					\node[red] at (17.5, -0.4) () {\Large$v_j$};
					\node[red] at (9.5, 1.5) () {\LARGE$V_i$};
					\node[red] at (18.5, 1.5) () {\LARGE$V_j$};
					
					\node[] at (3, -8) () {\huge$G_4$};
					\node[] at (14, -8) () {\huge$G_4'$};
					
					\draw[->, line width=1.5mm] (7.7,-3) -- (9.3,-3);
				\end{tikzpicture}
			}
			\caption{Our constructions for the proof of Theorem \protect \ref{thm:col_anon}, in the case of $k=4$. The gray strips indicate complete bipartite graphs. The edges we are interested in appear explicitly.}
			\label{fig:G, G'}
		\end{figure}
		
		We claim that, with the same certificate function $c$, all the vertices of $G_k'$ accept. Indeed, the only vertices whose neighborhood have been modified between $G_k$ and $G'_k$ are $u_i,v_i,u_j,v_j$. So every vertex in $V \setminus \{u_i,v_i,u_j,v_j\}$ accepts. For $u_i$, the only difference between its neighborhood in $G_k$ and $G'_k$ is that $u_j$ is replaced by $v_i$. Since $c(u_j)=c(v_i)$, the view of $u_i$ is the same in $G_k$ and $G'_k$, so $u_i$ accepts in $G'_k$ as well. Similarly, one can prove that $v_i$, $u_j$, and $v_j$ accept in $G'_k$.
		
		However, $G'_k$ is not $k$-colorable. Indeed, assume by contradiction that it is, and let $\varphi$ be a proper $k$-coloring of $G_k'$. For any $r \notin \{i,j\}$, let $w_r \in V_r$. Let $w_j$ be a vertex of $V_j \setminus \{u_j,v_j\}$ which exists, since each set contains at least $3$ vertices. Then, $K = \{w_1, \ldots, w_{i-1}, w_{i+1}, \ldots w_k\}$ is a clique in $G_k'$, so the $(k-1)$ vertices of $K$ receive pairwise different colors in $\varphi$. Moreover, both $u_i$ and $v_i$ are complete to $K$. So if $G_k'$ is $k$-colorable, then $u_i$ and $v_i$ have to be colored the same, which is a contradiction since $(u_i,v_i)\in E(G_k')$.
	\end{proof}
	
	Obtaining lower bounds is usually harder in the locally checkable proofs model. 
	In this more demanding model, we do not get the exact bound in terms of number certificates, but still we get a bound that would be considered optimal by the usual standards.
	Namely, we will prove that $\Omega(\log k)$ bits are needed to certify the $k$-colorability of a graph. More generally, we prove that, if vertices can see their neighborhoods at distance $d$, at least $\Omega(\log k /d)$ bits are needed to certify $k$-colorability. Note that the graph constructed in the proof of Theorem~\ref{thm:col_anon} has diameter two, thus no lower bound at distance at least $3$ can be obtained from that construction. 
	The lower bound of the following theorem is obtained with a completely different graph.
	
	\ThmLBColoringDistanced*
	
	\begin{proof}
		For $r, s \geqslant 2$, let us denote by $P_{r, s}$ the graph on vertex set $V = \{1, \ldots, r\} \times \{1, \ldots, s\}$ with the following edges:
		\begin{itemize}[label=--]
			\item for all $i \neq j \in \{1, \ldots, r\}$ and all $p \in \{1, \ldots, s\}$, $((i,p),(j,p)) \in E$
			\item for all $i \neq j \in \{1, \ldots, r\}$ and all $p \in \{1, \ldots, s-1\}$, $((i,p),(j,p+1)) \in E$
		\end{itemize}
		
		In other words, the graph $P_{r,s}$ is a sequence of $s$ cliques $\mathcal{K}_1, \ldots, \mathcal{K}_s$, each having size~$r$, with an antimatching\footnote{Remember that an antimatching between two sets of nodes is an edge set containing all the possible edges between the two sets, except for a matching.} between $\mathcal{K}_i$ and $\mathcal{K}_{i+1}$ for all $i\le s-1$. The second coordinate indicates the index of the clique the vertex belongs to, and for all $(i,p)$ with $p \le s-1$, $(i,p)$ is the only vertex in $\mathcal{K}_p$ which is not adjacent to $(i,p+1)$. For instance, the graph $P_{3,5}$ is depicted on Figure~\ref{fig:P_35}.
		
		\begin{figure}[h!]
			\centering
			\begin{tikzpicture}
				\node[circle, draw, minimum size=0.2cm] at (0, 0) (1) {};
				\node[circle, draw, minimum size=0.2cm] at (0, 2) (2) {};
				\node[circle, draw, minimum size=0.2cm] at (0.5, 1) (3) {};
				\node[circle, draw, minimum size=0.2cm] at (3, 0) (4) {};
				\node[circle, draw, minimum size=0.2cm] at (3, 2) (5) {};
				\node[circle, draw, minimum size=0.2cm] at (3.5, 1) (6) {};
				\node[circle, draw, minimum size=0.2cm] at (6, 0) (7) {};
				\node[circle, draw, minimum size=0.2cm] at (6, 2) (8) {};
				\node[circle, draw, minimum size=0.2cm] at (6.5, 1) (9) {};
				\node[circle, draw, minimum size=0.2cm] at (9, 0) (10) {};
				\node[circle, draw, minimum size=0.2cm] at (9, 2) (11) {};
				\node[circle, draw, minimum size=0.2cm] at (9.5, 1) (12) {};
				\node[circle, draw, minimum size=0.2cm] at (12, 0) (13) {};
				\node[circle, draw, minimum size=0.2cm] at (12, 2) (14) {};
				\node[circle, draw, minimum size=0.2cm] at (12.5, 1) (15) {};
				
				\draw
				(1) -- (2)
				(1) -- (3)
				(1) -- (5)
				(1) -- (6)
				(2) -- (3)
				(2) -- (4)
				(2) -- (6)
				(3) -- (4)
				(3) -- (5)
				(4) -- (5)
				(4) -- (6)
				(4) -- (8)
				(4) -- (9)
				(5) -- (6)
				(5) -- (7)
				(5) -- (9)
				(6) -- (7)
				(6) -- (8)
				(7) -- (8)
				(7) -- (9)
				(7) -- (11)
				(7) -- (12)
				(8) -- (9)
				(8) -- (10)
				(8) -- (12)
				(9) -- (10)
				(9) -- (11)
				(10) -- (11)
				(10) -- (12)
				(10) -- (14)
				(10) -- (15)
				(11) -- (12)
				(11) -- (13)
				(11) -- (15)
				(12) -- (13)
				(12) -- (14)
				(13) -- (14)
				(13) -- (15)
				(14) -- (15);
			\end{tikzpicture}
			\caption{The graph $P_{3,5}$. Here the cliques are triangles, and two consecutive cliques are linked by all possible edges except for the horizontal ones (the antimatching).}
			\label{fig:P_35}
		\end{figure}
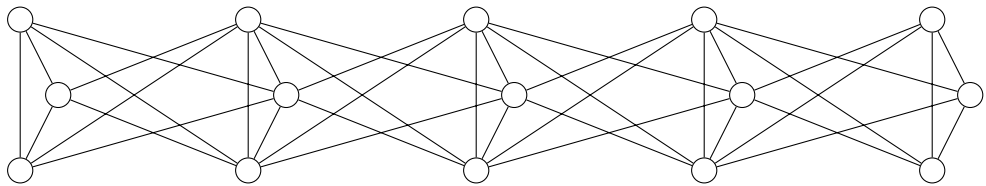
		
		Now, for any pair of permutations $\sigma, \tau$ of $\{1, \ldots, r\}$, let us denote by $G_{r, s}(\sigma, \tau)$ the graph obtained in the following way. We take two copies of $P_{r,s}$, and we denote their cliques by $\mathcal{K}_1, \ldots, \mathcal{K}_s$ and $\mathcal{K}'_1, \ldots, \mathcal{K}'_s$. We give $2rs$ different identifiers to the vertices of the two copies (the identifiers are fixed, they do not depend on $\sigma$ and $\tau$). Finally, we add the antimatching $\{(i, \sigma(i))\}_i$ between $\mathcal{K}_1$ and $\mathcal{K}_1'$, and the antimatching $\{(i, \tau(i))\}_i$ between $\mathcal{K}_s$ and $\mathcal{K}_s'$ (see Figure~\ref{fig:G_sigmatau} for an illustration).
		
		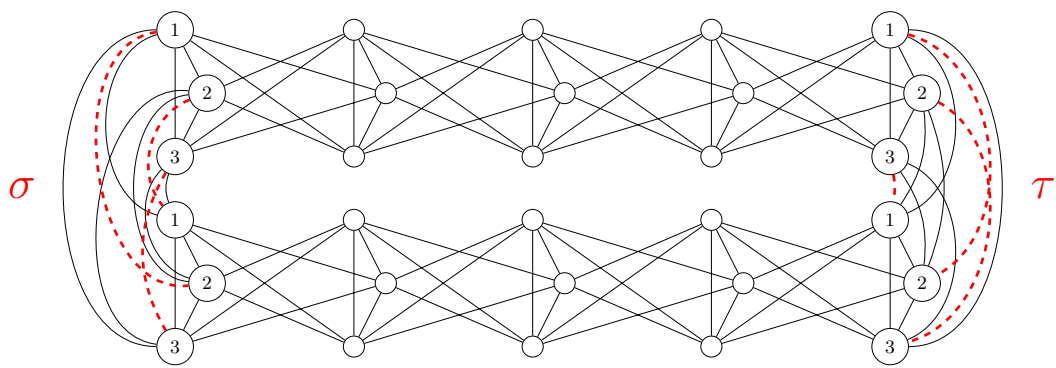
\begin{figure}[h]
			\centering
			\scalebox{0.84}{
				\begin{tikzpicture}
					\node[circle, draw, minimum size=0.2cm] at (0, 3) (1a) {\footnotesize 3};
					\node[circle, draw, minimum size=0.2cm] at (0, 5) (2a) {\footnotesize 1};
					\node[circle, draw, minimum size=0.2cm] at (0.5, 4) (3a) {\footnotesize 2};
					\node[circle, draw, minimum size=0.2cm] at (2.8, 3) (4a) {};
					\node[circle, draw, minimum size=0.2cm] at (2.8, 5) (5a) {};
					\node[circle, draw, minimum size=0.2cm] at (3.3, 4) (6a) {};
					\node[circle, draw, minimum size=0.2cm] at (5.6, 3) (7a) {};
					\node[circle, draw, minimum size=0.2cm] at (5.6, 5) (8a) {};
					\node[circle, draw, minimum size=0.2cm] at (6.1, 4) (9a) {};
					\node[circle, draw, minimum size=0.2cm] at (8.4, 3) (10a) {};
					\node[circle, draw, minimum size=0.2cm] at (8.4, 5) (11a) {};
					\node[circle, draw, minimum size=0.2cm] at (8.9, 4) (12a) {};
					\node[circle, draw, minimum size=0.2cm] at (11.2, 3) (13a) {\footnotesize 3};
					\node[circle, draw, minimum size=0.2cm] at (11.2, 5) (14a) {\footnotesize 1};
					\node[circle, draw, minimum size=0.2cm] at (11.7, 4) (15a) {\footnotesize 2};
					
					\node[circle, draw, minimum size=0.2cm] at (0, 0) (1b) {\footnotesize 3};
					\node[circle, draw, minimum size=0.2cm] at (0, 2) (2b) {\footnotesize 1};
					\node[circle, draw, minimum size=0.2cm] at (0.5, 1) (3b) {\footnotesize 2};
					\node[circle, draw, minimum size=0.2cm] at (2.8, 0) (4b) {};
					\node[circle, draw, minimum size=0.2cm] at (2.8, 2) (5b) {};
					\node[circle, draw, minimum size=0.2cm] at (3.3, 1) (6b) {};
					\node[circle, draw, minimum size=0.2cm] at (5.6, 0) (7b) {};
					\node[circle, draw, minimum size=0.2cm] at (5.6, 2) (8b) {};
					\node[circle, draw, minimum size=0.2cm] at (6.1, 1) (9b) {};
					\node[circle, draw, minimum size=0.2cm] at (8.4, 0) (10b) {};
					\node[circle, draw, minimum size=0.2cm] at (8.4, 2) (11b) {};
					\node[circle, draw, minimum size=0.2cm] at (8.9, 1) (12b) {};
					\node[circle, draw, minimum size=0.2cm] at (11.2, 0) (13b) {\footnotesize 3};
					\node[circle, draw, minimum size=0.2cm] at (11.2, 2) (14b) {\footnotesize 1};
					\node[circle, draw, minimum size=0.2cm] at (11.7, 1) (15b) {\footnotesize 2};
					
					\node at (-2.4, 2.5) () {\huge\textcolor{red}{$\sigma$}};
					\node at (13.6, 2.5) () {\huge\textcolor{red}{$\tau$}};
					
					\draw
					(1a) -- (2a)
					(1a) -- (3a)
					(1a) -- (5a)
					(1a) -- (6a)
					(2a) -- (3a)
					(2a) -- (4a)
					(2a) -- (6a)
					(3a) -- (4a)
					(3a) -- (5a)
					(4a) -- (5a)
					(4a) -- (6a)
					(4a) -- (8a)
					(4a) -- (9a)
					(5a) -- (6a)
					(5a) -- (7a)
					(5a) -- (9a)
					(6a) -- (7a)
					(6a) -- (8a)
					(7a) -- (8a)
					(7a) -- (9a)
					(7a) -- (11a)
					(7a) -- (12a)
					(8a) -- (9a)
					(8a) -- (10a)
					(8a) -- (12a)
					(9a) -- (10a)
					(9a) -- (11a)
					(10a) -- (11a)
					(10a) -- (12a)
					(10a) -- (14a)
					(10a) -- (15a)
					(11a) -- (12a)
					(11a) -- (13a)
					(11a) -- (15a)
					(12a) -- (13a)
					(12a) -- (14a)
					(13a) -- (14a)
					(13a) -- (15a)
					(14a) -- (15a)
					
					(1b) -- (2b)
					(1b) -- (3b)
					(1b) -- (5b)
					(1b) -- (6b)
					(2b) -- (3b)
					(2b) -- (4b)
					(2b) -- (6b)
					(3b) -- (4b)
					(3b) -- (5b)
					(4b) -- (5b)
					(4b) -- (6b)
					(4b) -- (8b)
					(4b) -- (9b)
					(5b) -- (6b)
					(5b) -- (7b)
					(5b) -- (9b)
					(6b) -- (7b)
					(6b) -- (8b)
					(7b) -- (8b)
					(7b) -- (9b)
					(7b) -- (11b)
					(7b) -- (12b)
					(8b) -- (9b)
					(8b) -- (10b)
					(8b) -- (12b)
					(9b) -- (10b)
					(9b) -- (11b)
					(10b) -- (11b)
					(10b) -- (12b)
					(10b) -- (14b)
					(10b) -- (15b)
					(11b) -- (12b)
					(11b) -- (13b)
					(11b) -- (15b)
					(12b) -- (13b)
					(12b) -- (14b)
					(13b) -- (14b)
					(13b) -- (15b)
					(14b) -- (15b)
					
					(1a) edge[bend right=20] node{} (2b)
					(1a) edge[bend right=60] node{} (3b)
					(2a) edge[bend right=90] node{} (1b)
					(2a) edge[bend right=75] node{} (2b)
					(3a) edge[bend right=90] node{} (1b)
					(3a) edge[bend right=85] node{} (3b)
					
					(13a) edge[bend left=68] node{} (13b)
					(13a) edge[bend left=20] node{} (15b)
					(14a) edge[bend left=90] node{} (13b)
					(14a) edge[bend left=68] node{} (14b)
					(15a) edge[bend left=20] node{} (14b)
					(15a) edge[bend left=20] node{} (15b)
					
					(1a) edge[bend right=30, red, very thick, dashed] node{} (1b)
					(2a) edge[bend right=90, red, very thick, dashed] node{} (3b)
					(3a) edge[bend right=55, red, very thick, dashed] node{} (2b)
					
					(13a) edge[bend left=10, red, very thick, dashed] node{} (14b)
					(14a) edge[bend left=68, red, very thick, dashed] node{} (15b)
					(15a) edge[bend left=68, red, very thick, dashed] node{} (13b);
					
				\end{tikzpicture}
			}
			\caption{The graph $G_{3,5}(\sigma,\tau)$, where $\sigma$ is the permutation $(1,2)$ and $\tau$ the cycle $(1,2,3)$.}
			\label{fig:G_sigmatau}
		\end{figure}
		
		\begin{claim}
			\label{claim:sigmatau}
			The graph $G_{r, s}(\sigma, \tau)$ is $r$-colorable if and only if $\sigma=\tau$.
		\end{claim}
		
		\begin{proof}
			Since $\mathcal{K}_1$ is a clique of size $r$, in a proper coloring with $r$ colors, every color appears exactly once inside $\mathcal{K}_1$. 
			Similarly, every color appears exactly once in a proper coloring of~$\mathcal{K}_2$. For all $i\in\{1, \ldots, r\}$, the vertex $(i,1)$ is a neighbor of every vertex in $\mathcal{K}_2$ except $(i,2)$. Thus, the vertices $(i,1)$ and $(i,2)$ are colored the same. Hence, given a coloring of $\mathcal{K}_1$, there exists a unique way to properly color $\mathcal{K}_2$. This propagates along the cycle of cliques $\mathcal{K}_2, \ldots \mathcal{K}_s, \mathcal{K}'_s, \ldots, \mathcal{K}'_1$, and it leads to a proper coloring of the whole graph if and only if the coloring of $\mathcal{K}'_1$ is compatible with the coloring of $\mathcal{K}_1$, that is, if and only if $\sigma=\tau$.
		\end{proof}
		
		Using Claim~\ref{claim:sigmatau}, we will deduce a lower bound on the number of bits needed to certify $k$-colorability, when vertices are allowed to see at distance $d$. 
		We consider the construction above with parameters $r=k$ and $s=2d$.
		Assume that $m$ bits are sufficient. Then, for any permutation $\sigma$, there exists a certificate function $c_\sigma : V \rightarrow \{0, \ldots, 2^m-1\}$ such that all the vertices of $G_{k,2d}(\sigma,\sigma)$ accept. 
		Suppose that there are two different permutations $\sigma, \tau$ such that the certification functions are the same: $c_\sigma=c_\tau$, that is, all vertices receive the same certificate in both instances.
		Then $G_{k,2d}(\sigma,\tau)$ would be accepted with this certificate function, since the view of every vertex would be identical as its view in $G_{k, 2d}(\sigma,\sigma)$ (for vertices in the left part) or in $G_{k, 2d}(\tau,\tau)$ (for vertices in the right part). This would be a contradiction because $G_{k, 2d}(\sigma,\tau)$ is not $k$-colorable (by Claim~\ref{claim:sigmatau}).
		
		Hence, all the certificate functions $c_\sigma$ are different. In particular, there are no more permutations of $\{1, \ldots, k\}$ than functions $V \rightarrow  \{0, \ldots, 2^m-1\}$. 
		The set $V$ has size $4kd$ since it is made of two copies of a graph of size $k\times 2d$. 
		Therefore, we get $k! \leqslant 2^{4mkd}$, leading to $m \geqslant \frac{\log_2(k!)}{4kd}$. Finally, since $\log_2(k!)=\Omega(k \log k)$, we get the result.
	\end{proof}
	
	\begin{remark}
		Two ingredients of the proof, the  communication complexity insight and the antimatchings to propagate the coloring, have already been used in~\cite{GoosS16}, to get a lower bound on the certification of not-3-colorable graphs.
	\end{remark}

	\subsection{Upper bounds: uniquely colorable graphs}
	
	We can wonder if the lower bound of Theorem~\ref{thm:LB_coloring_distance_d} can be improved and be independent of $d$. We will show that if it is the case, it is impossible with graphs similar to the ones used to prove Theorem~\ref{thm:LB_coloring_distance_d}. Indeed, these graphs have a strong property which can be exploited to get an almost tight upper bound: all the balls of radius $d \geqslant 2$ are uniquely colorable, up to the permutation of color classes.
	This property, that allowed us to control how the coloring can be transmitted and helped us to obtain the lower bound, can also be used to improve the upper bound. Let us first formally define uniquely colorable graphs.
	
	\begin{definition}
		Let $k, n \in \NN$. A graph $G$ is \emph{uniquely $k$-colorable at distance $d$} if, for every vertex~$v$, the graph induced by $B(v,d)$ is uniquely $k$-colorable up to the permutation of the color classes.
	\end{definition}
	
	One can easily note that, for every $r \geqslant 1, s \geqslant 3$, and for all permutations $\sigma, \tau$ of $\{1, \ldots, r\}$, the graphs $G_{r,s}(\sigma, \tau)$ introduced in the proof of Theorem~\ref{thm:LB_coloring_distance_d} are uniquely $r$-colorable at distance 2.
	Note moreover that every graph $G$ that is uniquely $k$-colorable at distance $d$ is either uniquely $k$-colorable at distance $d+1$, or has a ball of radius $d+1$ which is not $k$-colorable.
	
	Before proving the main statements of this section, let us first give some properties satisfied by uniquely $k$-colorable graphs:

	\begin{lemma}
		\label{lem:layers}
		Let $k, d \in \NN$.
		For every uniquely $k$-colorable graph $G$ at distance $d$ which is (globally) $k$-colorable, every color~$\alpha$ and every vertex~$u$, the following holds:
		\begin{enumerate}[label=(\roman*)]
			\item at least one vertex of color $\alpha$ appears in $N[u]$;
			\item for every $i \in \NN$, at least one vertex of color $\alpha$ appears at distance $i$, $i+1$ or $i+2$ from~$u$ (if $N^{i+2}(u)\neq\emptyset$).
		\end{enumerate}
	\end{lemma}
	
	\begin{proof}
		Let  $G$ be a uniquely $k$-colorable graph at distance $d$, $u$ be a vertex and $\alpha$ be a color.
		\begin{enumerate}[label=(\roman*)]
			\item Assume by contradiction that color $\alpha$ does not appear in the closed neighborhood of~$u$. Then, we can obtain another proper coloring by changing the color of $u$ to $\alpha$, which is a contradiction since $G$ is uniquely $k$-colorable.
			\item Let $v$ be a vertex at distance $i+1$ from $u$. By applying~(i) to $v$, there exists a vertex $w \in N[v]$ of color $\alpha$, and $w$ is at distance $i$, $i+1$ or $i+2$ from $u$.
			\qedhere
		\end{enumerate}
	\end{proof}

	We are now ready to state and prove the following theorem.
	
	\ThmUniquelyColorable*
	
	\begin{proof}
		We actually prove the following more precise result in the model where vertices are allowed to see at distance $d\geqslant 11$. Let $\delta := \left\lfloor \frac{d-2}{9} \right\rfloor$ and $f(k):=\left\lceil \sqrt[\delta]{k} \right\rceil$. Note that $\delta \geqslant 1$, and that every integer in $\{0, \ldots k-1\}$ needs at most $\delta$ digits to be represented in base $f(k)$. Let us prove that $f(k)+1$ certificates are sufficient to certify $k$-colorability for uniquely $k$-colorable graphs at \mbox{distance $d-2$}. (The theorem follows by taking the logarithm of this quantity.)
		The set of certificates we will use is $\{0, 1, \ldots, f(k)\}$.
		
		Let $G=(V,E)$ be a uniquely $k$-colorable graph at distance $d-2$ which is $k$-colorable. Let us denote by $\varphi$ a proper $k$-coloring of $G$ using colors $\{0, \ldots, k-1\}$. The prover gives the following certificates. 
		The prover first chooses a maximal independent set $X$ at distance $6\delta+1$, and gives to all the vertices $u \in X$ the certificate $c(u):=f(k)$. Then, for every vertex $v \in V \setminus X$ such that $d(v,X) \leqslant 3\delta$, let $i_v:=\left\lceil \frac{d(v,X)}{3}\right\rceil$. The certificate $c(v)$ given by the prover to $v$ is the $i_v$-th digit of the decomposition of $\varphi(v)$ in base $f(k)$. Note that $c(v)\neq f(k)$. Finally, the prover gives to all the vertices at distance at least $3\delta+1$ from $X$ an arbitrary certificate in $\{0, \ldots, f(k)-1\}$.
		Intuitively, the name of the color is encoded on the vertices of that color, the first bit on a first circular strip (of thickness three:  the vertices at distance 1, 2 and 3 from the vertex in $X$), the second bit on a second circular strip, etc.
		
		The informal idea of the verification is the following. Since $X$ is a maximal independent set at distance $6\delta+1$, each vertex $v$ should see at least one vertex $u \in X$ at distance $6\delta\leqslant d-2$. Since $G$ is uniquely colorable, $v$ will determine the set of vertices around $u$ which are in its color class, and will then determine its own color thanks to their certificates. Finally, $v$ will compare its color to the color of its neighbors.
		
		Let us now formalize this intuition. Every vertex $v \in V$ checks its certificate as follows. If the diameter of the connected component of $v$ is at most $6\delta$, then $v$ sees its whole connected component and accepts if and only if it is $k$-colorable. Otherwise, let us denote by $X$ the set of vertices with certificate $f(k)$. The vertex $v$ performs the following verification steps:
		
		\begin{enumerate}[label=(\roman*)]
			\item  $v$ starts by checking if $B(v,6\delta)\cap X \neq \emptyset$. If it is not the case, then $v$ rejects.
			
			\item Let $u \in B(v, 6\delta) \cap X$. Since $d\geqslant 9\delta+2$, $B(v,d-2)$ contains $B(u,3\delta)$.
			Since $B(v,d-2)$ is uniquely colorable, $v$ determines the color classes $\mathcal{C}_1, \ldots, \mathcal{C}_k$ of $B(v,d-2)$. Let us denote by $\mathcal{C}_\ell$ the color class of~$v$. For every $i \in \{1, \ldots, \delta\}$, $v$ checks if all the vertices in $\mathcal{C}_\ell \cap (N^{3i-2}(u) \cup N^{3i-1}(u) \cup N^{3i}(u))$ have the same certificate. If it is not the case, $v$ rejects. Otherwise, let us denote by $a_i$ this common certificate. Note that $a_i$ is well-defined. Indeed, since the connected component of $v$ has diameter at least $6\delta+1$, $N^j(u)\neq\emptyset$ for all $j\leqslant 3\delta$. And Lemma~\ref{lem:layers} ensures that at least one vertex of $\mathcal{C}_\ell$ appears in every three consecutive layers around $u$. Let $a(v,u)$ be the integer whose decomposition in base $f(k)$ is $a_1\ldots a_\delta$.
			
			\item  Then, $v$ checks if, for every pair of vertices $u,u' \in B(v,6\delta)\cap X$, we have $a(v,u)=a(v,u')$. If it is not the case, $v$ rejects. Otherwise, let us denote by $a(v)$ this common value. If $a(v) \notin \{0, \ldots, k-1\}$, then $v$ rejects.
			
			\item For every $w\in N(v)$, $v$ can determine $a(w)$. Indeed, $a(w)$ only depends on $B(w,d-2)$ which is included in $B(v,d-1)$. So it is in the view of $v$. If $a(w)=a(v)$ for some $w \in N(v)$, then $v$ rejects.
			
			\item If $v$ did not reject at this point, then $v$ accepts.
		\end{enumerate}
		
		To conclude, we simply have to show that a uniquely \mbox{$k$-colorable} graph at distance $d-2$ is accepted with this verification algorithm if and only if it is $k$-colorable.
		If a graph $G$ is $k$-colorable, then the above checking algorithm indeed accepts with the certificates given by the prover as described above.
		Conversely, if all the vertices accept, then we can construct a proper $k$-coloring of~$G$ by giving to each vertex $v$ the color $a(v)$. It uses at most $k$ colors because of step~(iii). Moreover, it is a proper coloring since step~(iv) ensures that every vertex~$v$ is colored differently from all its neighbors.
	\end{proof}
	
	\begin{remark}
		Unique colorability can also be tested by the verification algorithm (instead of making the assumption that the input graph is uniquely colorable). In other words, there is a certification using the same number of certificates, which accepts the input graph if and only if it is a uniquely $k$-colorable graph at distance $d-2$ which is $k$-colorable.
	\end{remark}

	One can wonder how much we can decrease the number of certificates when $d$ increases. We prove that we can decrease it up to 2 (which is the best one could hope for). 
	To do so, the main difficulty is that, in the proof of Theorem~\ref{thm:uniquely_colorable}, we used a special certificate (namely the certificate $f(k)$) to certify the vertices of the maximal independent set $X$, and then at least 2 other certificates to code the colors. We will show that we can get rid of the special certificate, by modifying the certification around the vertices of $X$.
	
	\ThmTwoCertificates*
	
	\begin{proof}
		Let $\delta = \lceil\log_2 k\rceil$. We will prove that, in the model where vertices can see at distance \mbox{$d=9\delta+8$}, only $2$ certificates are sufficient in order to certify that a uniquely $k$-colorable graph at distance $d-2$ is $k$-colorable. In the following, we denote these two certificates by 0 and 1.
		
		On a graph $G=(V,E)$ which is $k$-colorable, the prover gives the following certificates. Let us denote by $\varphi$ a proper $k$-coloring of $G$ using colors $\{0, \ldots, k-1\}$. The prover begins by choosing a maximal independent set $X$ at distance $6\delta+5$, and gives to all the vertices $u \in X$ the certificate 1. He also gives certificate 1 to all the vertices $v$ such that $d(v,X)=1$, and certificate 0 to all the vertices $v$ such that $d(v,X)=2$. For every vertex~$v$ such that $3 \leqslant d(v,X) \leqslant 3\delta+2$, let $i_v:=\left\lceil\frac{d(v,X)-2}{3}\right\rceil$. The certificate $c(v)$ given by the prover to $v$ is the $i_v$-th digit of the decomposition of $\varphi(v)$ in base 2. Finally, the prover gives certificate $0$ to all the vertices $v$ such that $d(v,X) \geqslant 3\delta+3$.
		
		A vertex $v \in V$ checks its certificate as follows. If the diameter of the connected component of $v$ in $G$ is at most $6\delta+4$, then $v$ sees its whole connected component and accepts if and only if it is $k$-colorable. Otherwise, let us denote by $Y$ the set of vertices $u$ such that (a) $c(u)=1$ and, (b) $c(u')=1$ for all $u' \in N(u)$ and, (c) $c(u')=0$ for all $u' \in N^2(u)$. The vertex $v$ performs the following verification steps:
		\begin{enumerate}[label=(\roman*)]
			\item $v$ starts by checking if $B(v,6\delta+4) \cap Y = \emptyset$. If it is the case, then $v$ rejects.
			
			\item Let $u\in B(v,6\delta+4)\cap Y$. Since $d=9\delta+8$, $B(v,d-2)$ contains $B(u,3\delta+2)$. Since $B(v,d-2)$ is uniquely colorable, $v$ determines the color classes $\mathcal{C}_1, \ldots, \mathcal{C}_k$ of $B(v,d-2)$. Let us denote by $\mathcal{C}_\ell$ the color class of~$v$. For every $i \in \{1, \ldots, \delta\}$, $v$ checks if all the vertices in $\mathcal{C}_\ell \cap (N^{3i}(u) \cup N^{3i+1}(u) \cup N^{3i+2}(u))$ have the same certificate. If it is not the case, $v$ rejects. Otherwise, let us denote by $a_i$ this common certificate. As in proof of Theorem~\ref{thm:uniquely_colorable}, $a_i$ is well-defined by Lemma~\ref{lem:layers}. Let $a(v,u)$ be the integer whose decomposition in base 2 is $a_1\ldots a_\delta$.
			
			\item  Then, $v$ checks if, for every pair of vertices $u, u' \in B(v,6\delta+4)\cap Y$, we have $a(v,u)=a(v,u')$. If it is not the case, $v$ rejects. Otherwise, let us denote by $a(v)$ this common value. If $a(v) \notin \{0, \ldots, k-1\}$, then $v$ rejects.
			
			\item For every $w \in N(v)$, $v$ can determine $a(w)$. Indeed, $a(w)$ only depends on $B(w,d-2)$ which is included in $B(v,d-1)$. So it is in the view of $v$. If $a(w)=a(v)$ for some $w \in N(v)$, then $v$ rejects.
			
			\item If $v$ did not reject at this point, then $v$ accepts.
		\end{enumerate}
		
		To conclude, we have to show that a uniquely $k$-colorable graph at distance $d-2$ is accepted with this verification algorithm if and only if it is $k$-colorable. If all the vertices accept, we construct a proper $k$-coloring of the graph by giving to each vertex $v$ the color $a(v)$. It uses at most $k$ colors because of step~(iii) and it is a proper coloring because step~(iv) ensures that every vertex $v$ is colored differently from all its neighbors.
		
		Conversely, assume that a graph $G$ is $k$-colorable. We need to show that every vertex accepts with the certificate function given by the prover as described before. It is more difficult than in proof of Theorem~\ref{thm:uniquely_colorable}, since we have $X\subseteq Y$ but we may not have the equality. Hence, some vertices in $Y\setminus X$ can lead some vertex to reject at step (iii). We will prove that it cannot happen because of the following claim:
		
		\begin{claim}
			\label{claim:special_vertex}
			With the certificate function given by the prover, every vertex $y \in Y\setminus X$ is a twin of some vertex $x \in X$ (that is $N[y]=N[x]$).
		\end{claim}
		\begin{proof}
			Let us consider the certificates assigned by the prover, and let $y \in Y\setminus X$. Let us denote by $x$ the closest vertex in $X$ from $y$. By construction, we have $X \subseteq Y$, so $x \in Y$. First, let us prove that $d(x,y)=1$. Assume by contradiction that $d(x,y) \geqslant 2$. We distinguish some cases depending on $d(x,y)$:
			\begin{itemize}
				\item $d(x,y)=2$. Since $y\in Y$, we get $c(x)=0$. But $x\in Y$, so it should have certificate 1, which is a contradiction.
				\item $d(x,y)=3$. Then, there exists $w \in N^2(y)\cap N(x)$. This is a contradiction, since $y \in Y$ implies that $c(w)=0$ and $x \in Y$ implies that $c(w)=1$.
				\item $4 \leqslant d(x,y) \leqslant 3\delta+2$. By Lemma~\ref{lem:layers}, there is a vertex $z\in N[y]$ colored by~0. Note that $d(z,X)\geqslant 3$. \\
				If $d(z,X)\leqslant 3\delta+2$, then the certificate of $z$ given by the prover is a digit of the decomposition of its color in base $2$. Since $z$ is colored by $0$, $z$ has certificate $0$. \\
				If $d(z,X) \geqslant 3\delta+3$, then $z$ has certificate $0$ by construction.\\
				So $y$ has a neighbor with certificate $0$, which is a contradiction since $y\in V$.
				\item $d(x,y) \geqslant 3\delta+3$. Then, $y$ has certificate $0$ by construction and $y$ cannot be in $Y$.
			\end{itemize}
			So $d(x,y)=1$. Let us finally prove that $x$ and $y$ are twins. We prove that $N(y)\setminus\{x\}=N(x)\setminus\{y\}$. Let $z \in N(y)\setminus\{x\}$. Since $y \in Y$, we get $c(z)=1$. Moreover, $d(x,z) \in \{1,2\}$ since $d(x,y)=1$. Since $x \in Y$, we cannot have $d(x,z)=2$ since otherwise this would imply $c(z)=0$. Thus, $d(x,z)=1$, so $z \in N(x) \setminus\{y\}$. Thus, $N(y)\setminus\{x\}\subseteq N(x)\setminus\{y\}$. By symmetry, we get the other inclusion. So $x$ and $y$ are twins, which completes the proof of Claim~\ref{claim:special_vertex}.
		\end{proof}
		
		We can now conclude the proof of Theorem~\ref{thm:2_certificates_are_sufficient}. In the way the prover assigned certificates, every vertex $v$ will see at least one vertex $u \in Y$ at distance at most $6\delta+4$. It will then compute $a(v,u)$. But since $u$ is either in $X$ or a twin of a vertex $w \in X$, $u$ and $w$ has the same neighborhoods at distance $d$ for every $d \geqslant 2$. Thus, $a(v,u)=a(v,w)$. Then, by construction, every vertex $v$ will accept.
	\end{proof}

	\section{Dominating sets at large distance}
	
	In this section, we consider \emph{labeled} graphs: every graph~$G$ is equipped with a label function $\mathcal{L} : V(G) \rightarrow \{0,1\}$ which is part of the input (in other words, the certificate function given by the prover depends on the label function). We say that a vertex $v$ is \emph{labeled} if $\mathcal{L}(v)=1$. We focus on the number of certificates needed to certify that the set of labeled vertices is a \emph{dominating set at distance $t$} (we remind that a dominating set at distance $t$ is a set $S\subseteq V$ such that $\cup_{v\in S}B(v,t)$ is equal to $V$).
	
	\subsection{Lower bounds}
	
	\ThmSqrtNeeded*
	
	\begin{proof}
		Let us denote by $P_{t+1}$ the path having $t+1$ vertices denoted by $u_0, \ldots, u_t$, where only $u_0$ is labeled. Note that $\{u_0\}$ is a dominating set at distance $t$ in $P_{t+1}$ (see Figure~\ref{fig:g6} for an illustration).
		
		\begin{figure}[h]
			\centering
			\begin{tikzpicture}
				\node[circle, draw, minimum size=0.7cm, fill=red!50] at (0, 0) (0) {\scriptsize $u_0$};
				\node[circle, draw, minimum size=0.7cm] at (1.5, 0) (1) {\scriptsize $u_1$};
				\node[circle, draw, minimum size=0.7cm] at (3, 0) (2) {\scriptsize $u_2$};
				\node[circle, draw, minimum size=0.7cm] at (4.5, 0) (3) {\scriptsize $u_3$};
				\node[circle, draw, minimum size=0.7cm] at (6, 0) (4) {\scriptsize $u_4$};
				\node[circle, draw, minimum size=0.7cm] at (7.5, 0) (5) {\scriptsize $u_5$};
				\node[circle, draw, minimum size=0.7cm] at (9, 0) (6) {\scriptsize $u_6$};
				
				\draw
				(0) -- (1)
				(1) -- (2)
				(2) -- (3)
				(3) -- (4)
				(4) -- (5)
				(5) -- (6)
				;
			\end{tikzpicture}
			\caption{The graph $P_7$.}
			\label{fig:g6}
		\end{figure}
		
		Since $P_{t+1}$ is a valid instance, there exists a certificate function $c$ such that every vertex accepts.
		Let us prove that for $1\le i < j \le t-1$, the pairs of certificates $(c(u_i),c(u_{i+1}))$ and $(c(u_j),c(u_{j+1}))$ must be distinct. 
		Assume by contradiction that it is not the case, and let $i < j$ such that $(c(u_i),c(u_{i+1}))=(c(u_j),c(u_{j+1}))$. Let us prove that, for every possible pair $(i,j)$, we can find a graph $G$ that is accepted and should not be. There are several cases:
		\begin{itemize}
			\item if $j=i+1$, we get $c(u_i)=c(u_{i+1})=c(u_{i+2})$. Thus, the triangle where all the vertices receive the certificate $c(u_{i+1})$ is accepted, since the view of every vertex of the triangle is the same as the view of $u_{i+1}$ in $P_{t+1}$. But this graph should not be accepted since it does not contain any labeled vertex.
			
			\item if $j=i+2$, then the cycle of length 4 where the vertices receive successively the certificates $c(u_i)$ and $c(u_{i+1})$ is accepted, since every vertex has either the view of $u_{i+1}$ or the view of $u_{i+2}$ in~$P_{t+1}$. But it should not be accepted since it does not contain any labeled vertex.
			
			\item if $j \geqslant i+3$, then the cycle of length $j-i$ where vertices have certificates $c(u_i)$, $ c(u_{i+1}), \ldots,$ $c(u_{j-2}), c(u_{j-1})$ is accepted. Indeed, for every $\ell \in \{i+1, \ldots, j-1\}$, the vertex certified by $c(u_\ell)$ has the view of $u_\ell$ in $P_{t+1}$, and the vertex certified by $c(u_i)$ has the view of $u_j$ in~$P_{t+1}$. Again, this graph should not be accepted since it does not contain any labeled vertex.
		\end{itemize}
		Hence, all the pairs of certificates of consecutive vertices have to be pairwise distinct. Since we need at least $t-1$ different pairs of certificates, there must be at least $\sqrt{t-1}$ different certificates in $P_{t+1}$.
	\end{proof}
	
	We can wonder if the lower bound of Theorem~\ref{thm:sqrt_needed} can be generalized when vertices are allowed to see at distance $d \leqslant t$. We will see that in Corollary~\ref{cor:LB_dominating} that the lower bound is in fact divided by~$d$.

	\ThmLBDominating*
	
	\begin{corollary}
		\label{cor:LB_dominating}
		In the anonymous model where vertices can see at distance $d \leqslant \frac{t}{3}$, $\Omega(\log(t)/d)$ bits are needed to certify a dominating set at distance $t$.
	\end{corollary}
	
	\begin{proof}[Proof of Theorem~\ref{thm:LB_dominating}]
		The proof generalizes the one of Theorem~\ref{thm:sqrt_needed}. As in the proof of Theorem~\ref{thm:sqrt_needed}, we consider the graph $P_{t+1}$. Since $P_{t+1}$ has a dominating set at distance $t$, there exists a certificate function $c$ such that every vertex accepts. Let us prove that $2d$-tuples of consecutive certificates $(c(u_i),\ldots,c(u_{i+2d-1}))$ have to be pairwise distinct for every \mbox{$i \in \{1, \ldots, t-2d+1\}$}. Assume by contradiction there exist $i<j$ such that $(c(u_i), \ldots, c(u_{i+2d-1})) = (c(u_j), \ldots, c(u_{j+2d-1}))$. Let $r=j-i$. The finite sequence $(c(u_i), c(u_{i+1}), \ldots, c(u_{j+2d-1}))$ has its $2d$ first and last values equal, so it is $r$-periodic.\footnote{It is not necessarily a cyclic periodicity. By saying that this finite sequence is $r$-periodic, we mean that $c(u_k)=c(u_{k+r})$ for all $k \in \{i, \ldots, i+2d-1\}$.} Let us consider the cycle $C$ on $r \cdot (2d+1)$ vertices denoted by $v_0, \ldots, v_{r(2d+1)-1}$ with no labeled vertex. The graph $C$ does not have any labeled vertex, so it should not be accepted. However, let us show that it is accepted with some certificate function. For every $\ell \in \{0, \ldots, r\cdot (2d+1)-1\}$, the prover gives to the vertex $v_\ell$ the certificate $c(u_{i+d+(\ell\mod r)})$. We claim that every vertex accepts. Indeed, for every $\ell \in \{0, \ldots, r\cdot(2d+1)-1\}$, $v_\ell$ has the same view as $u_{i+d+(\ell\mod r)}$ in $P_{t+1}$. This is a contradiction.
		
		Hence, at least $t-2d+1$ different $2d$-tuples of certificates are needed to certify $P_{t+1}$, so at least $\sqrt[2d]{t-2d+1}$ different certificates are needed.
	\end{proof}

	\subsection{Upper bounds}
	
	In this subsection, we prove upper bounds matching with the previous lower bounds, thus giving the optimal bounds to certify a dominating set at large distance.
	Firstly, the lower bound of Theorem~\ref{thm:sqrt_needed} is optimal up to a constant factor, as stated in the following Theorem~\ref{thm:sqrt_sufficient}, which shows that the optimal size is $\frac{\log t}{2} + o(\log t)$ if the vertices can see at distance~1.
	
	\ThmSqrtSufficient*
	
	Before proving Theorem~\ref{thm:sqrt_sufficient}, we will define a few notions. 
	Let $A$ be an alphabet, and $\omega, \omega' \in A^\ast$. We say that $\omega$ is a \emph{factor} of $\omega'$ if there exists a sequence of consecutive letters in $\omega'$ which is equal to $\omega$. Let us now define \emph{de Bruijn words}, whose existence is well-known.
	
	\begin{proposition}
		\label{prop:de_bruijn}
		Let $k, n \in \NN^\ast$, and $A$ be an alphabet of size $k$. There exists a word $\omega \in A^\ast$ of length $k^n$, such that every word of $A^n$ appears at most once as a factor of $\omega$. Such a word $\omega$ is called a $(k,n)$-\emph{de Bruijn word}.\footnote{Usually, de Bruijn words are defined as words of length $k^n$ such that each word of $A^n$ appears at least once when $\omega$ is considered circularly. But due to the length of $\omega$, each word of $A^n$ actually appears exactly once circularly, so at most once if $\omega$ is not seen as a circular word.}
	\end{proposition}
	
	\begin{proof}[Proof of Theorem~\ref{thm:sqrt_sufficient}]
		Let $\tau = \left\lceil \sqrt{t} \right\rceil$. Let us prove that $3\tau$ certificates are sufficient to certify a dominating set at distance $t$ (in the anonymous model, where vertices can see at distance~$1$). Let $A := \{1, \ldots, \tau\}$. The certificates used in the scheme will be pairs in $C := \{0,1,2\} \times A$. For $(x,y) \in C$, let $\pi_1(x,y):=x$ and $\pi_2(x,y):=y$. Let $\omega'\in A^\ast$ be a $(\tau,2)$-de Bruijn word (which, by definition, has length at least~$t$), and let us denote by $\omega = \omega_1\ldots\omega_t$ its prefix of length exactly $t$.

		Let $G = (V,E)$ be a graph and $S \subseteq V$ be the set of labeled vertices.
		If $S$ is dominating at distance $t$, the certificate function given by the prover is the following. The vertices of $S$ are given an arbitrary certificate, and for every $u \in V \setminus S$ at distance $i$ from $S$, the prover gives to $u$ the certificate $c(u)$ which is $(i\!\!\mod 3,~\omega_i)$.
		
		The informal idea of the verification is the following one. In its certificate, every vertex $u$ is given a letter of $\omega$. By looking at its neighbors, $u$ will be able to determine its position in~$\omega$ (since a pair of letters defines a unique position in the de Bruijn word $\omega$), which corresponds to its distance to $S$.
		
		More formally, let $c$ be a certificate function. Each vertex $u$ checks the certificate as follows~:
		\begin{enumerate}[label=(\roman*)]
			\item If $N[u] \cap S \neq \emptyset$, then $u$ accepts.\footnote{We could verify that $\pi_2(c(u))$ corresponds to a letter in the beginning of $\omega$, but it is not necessary.}
			\item Else, $u$ checks that, for all $u',u'' \in N[u]$, if $\pi_1(c(u'))=\pi_1(c(u''))$ then $\pi_2(c(u'))=\pi_2(c(u''))$. If it is not the case, $u$ rejects.
			\item Then, $u$ checks if it has at least one neighbor $v$ such that \mbox{$\pi_1(c(v)) = \pi_1(c(u))-1 \mod 3$}, and if $\pi_2(c(v))\pi_2(c(u))$ is a factor of $\omega$. If it is not the case, $u$ rejects.
			\item Finally, for every neighbor $w$ such that \mbox{$\pi_1(c(w)) = \pi_1(c(u))+1 \mod 3$}, $u$ checks if $\omega$ has $\pi_2(c(v))\pi_2(c(u))\pi_2(c(w))$ as a factor. If it is not the case, then $u$ rejects.
			\item If $u$ did not reject at this point, it accepts.
		\end{enumerate}
		
		It remains to show that there exists a certificate function such that all the vertices of $G$ accept if and only if $S$ is dominating at distance $t$.
		If $S$ is a dominating set at distance $t$, then one can easily check that all the vertices accept with the certificates assigned by the prover as described previously. Note that with this certificate function, for step (ii), if $u', u'' \in N[u]$ satisfy $\pi_1(c(u'))=\pi_1(c(u''))$, then $d(u',S)=d(u'',S)$ so $\pi_2(c(u'))=\pi_2(c(u''))$.
		
		For the converse, assume that $G$ is accepted with some certificate function $c$. By (iii), every $u\in V$ such that no vertex is labeled in $N[u]$ should have a neighbor $v$ such that \mbox{$\pi_1(c(v)) = \pi_1(c(u))-1 \mod 3$} and $\pi_2(c(v))\pi_2(c(u))=\omega_\ell\omega_{\ell+1}$ for some $\ell$. Note that since $\omega$ is a de Bruijn word, this $\ell$ is unique. Let us prove by induction on $\ell \in \{1, \ldots, t-1\}$ that $d(u,S) \leqslant \ell+1$.
		\begin{itemize}
			\item For $\ell=1$, we have $\pi_2(c(v))\pi_2(c(u)) = \omega_1\omega_2$. Let us prove that $d(u,S) \leqslant 2$. It is sufficient to prove that $d(v,S) \leqslant 1$. If $v$ accepts at step (i), the conclusion holds. So we can assume that $v$ accepts at step~(v). By~(iii), $v$ has a neighbor~$v'$ such that \mbox{$\pi_1(c(v')) = \pi_1(c(v))-1 \mod 3$} and $\pi_2(c(v'))\pi_2(c(v))$ is a factor of~$\omega$. Since $v$ does not reject at step (iv), $\pi_2(c(v'))\pi_2(c(v))\pi_2(c(u))$ is a factor of $\omega$. So the letters $\omega_1\omega_2=\pi_2(c(v))\pi_2(c(u))$ appear at least twice as a factor of $\omega$, which is a contradiction with the definition of de Bruijn words. Indeed, $\pi_2(c(v))\pi_2(c(u))$ are already the first two letters of $\omega$. So if a factor $\pi_2(c(v'))\omega_1\omega_2$ appears somewhere, $\omega_1\omega_2$ must appear at least twice. Hence, $v$ accepts at step (i), so $d(v,S) \leqslant 1$.
			\item Assume now that $\ell \geqslant 2$. To prove that $d(u,S) \leqslant \ell+1$, it is sufficient to prove that $d(v,S) \leqslant \ell$. If $v$ accepts at step (i) the conclusion holds. So we can assume that $v$ accepts at step~(v). By (iii), $v$ has a neighbor~$v'$ such that \mbox{$\pi_1(c(v')) = \pi_1(c(v))-1 \mod 3$} and $\pi_2(c(v'))\pi_2(c(v))$ is a factor of~$\omega$. Since $v$ does not reject at step (iv), $\pi_2(c(v'))\pi_2(c(v))\pi_2(c(u))$ is a factor of~$\omega$. Since the factor $\pi_2(c(v))\pi_2(c(u))$ appears exactly once in $\omega$ and $\pi_2(c(v))\pi_2(c(u))=\omega_\ell\omega_{\ell+1}$, we have $\pi_2(c(v'))\pi_2(c(v)) = \omega_{\ell-1}\omega_\ell$ by~(iv). By induction hypothesis, it implies that $d(v,S) \leqslant \ell$.
		\end{itemize}
		Thus, for any vertex $u \in V$, if $u$ accepts at step (i) we have $d(u,S) \leqslant 1$, otherwise we have $d(u,S) \leqslant t$ by the previous induction. So $S$ is indeed a dominating set at distance $t$.
	\end{proof}

	In the case where vertices can see at distance $d \geqslant t$, we will see in Corollary~\ref{cor:UB_dominating} that the bound of Corollary~\ref{cor:LB_dominating} is also optimal: the number of bits is divided by $d$.

	\ThmUBDominating*
	
	\begin{corollary}
		\label{cor:UB_dominating}
		In the anonymous model where vertices can see at distance $d$, $O(\log(t)/d)$ bits are sufficient to certify a dominating set at distance $t$.
	\end{corollary}

	\begin{proof}[Proof of Theorem~\ref{thm:UB_dominating}]
		The idea of the proof is similar to the one of Theorem~\ref{thm:sqrt_sufficient}.
		
		Let $\tau = \left\lceil \sqrt[d+1]{t} \right\rceil$. Let us prove that $3\tau$ certificates are sufficient to certify a dominating set at distance $t$ (in the anonymous model, where vertices can see at distance $d$). Let $A = \{1, \ldots, \tau\}$. The set of certificates that will be given to vertices consists of pairs in $C = \{0,1,2\} \times A$. For $(x,y) \in C$, let $\pi_1(x,y):=x$ and $\pi_2(x,y):=y$. Let $\omega'\in A^\ast$ be a $(\tau,d+1)$-de Bruijn word (which, by definition, has length at least~$t$), and let us denote by $\omega = \omega_1\ldots\omega_t$ its prefix of length exactly $t$.

		Let $G = (V,E)$ be a graph and $S \subseteq V$ be the set of labeled vertices.
		If $S$ is a dominating set at distance~$t$, the certificate function given by the prover is the following. Vertices of $S$ are given an arbitrary certificate, and for every $u \in V \setminus S$ at distance $i$ from $S$, the certificate $c(u)$ is $(i\mod 3,~\omega_i)$.
		
		The informal idea of the verification is the following one. In its certificate, every vertex $u$ is given a letter of $\omega$. By looking at its neighbors, $u$ will be able to determine its position in~$\omega$, which corresponds to its distance to $S$.
		
		More formally, let $c$ be a certificate function. Each vertex $u_0$ checks the certificate as follows~:
		\begin{enumerate}[label=(\roman*)]
			\item If $B(u_0,d) \cap S \neq \emptyset$, then $u_0$ accepts.
			\item Else, $u_0$ checks if there exists a path $u_0u_1\ldots u_d$ such that $\pi_1(c(u_i)) = \pi_1(c(u_0))-i\mod 3$ for every $i \in \{0, \ldots, d\}$. Let us call such a path a \emph{decreasing path (starting at $u_0$)}. If $u_0$ does not have any decreasing path, $u_0$ rejects.
			\item Then, $u_0$ checks if $\pi_2(c(u_d))\ldots\pi_2(c(u_0))$ is the same for all the decreasing paths, and if it is a factor of $\omega$. If it is not the case, $u_0$ rejects.
			\item Finally, for every neighbor $w$ such that \mbox{$\pi_1(c(w)) = \pi_1(c(u_0))+1 \mod 3$}, $u_0$ checks if $\omega$ has $\pi_2(c(u_d))\ldots\pi_2(c(u_0))\pi_2(c(w))$ as a factor. If it is not the case, then $u_0$ rejects.
			\item If $u_0$ did not reject at this point, it accepts.
		\end{enumerate}
		
		It remains to show that there exists a certificate function such that all the vertices of $G$ accept if and only if $S$ is a dominating set at distance $t$.
		If $S$ is a dominating set at distance $t$, then one can easily check that all vertices accept with the certificates assigned by the prover described previously. Note that with this certificate function, for step (iii), if $u_0\ldots u_d$ is a decreasing path, then $d(u_i,S)=d(u_0,S)-i$ for all $i \in \{0, \ldots, d\}$, so $\pi_2(c(u_d))\ldots\pi_2(c(u_0))$ is a factor of $\omega$.
		
		For the converse, assume that $G$ is accepted with some certificate function $c$. By~(ii) and~(iii), every $u_0\in V$ such that no vertex is labeled in $B(u_0, d)$ should have a decreasing path $u_0u_1\ldots u_d$, such that $\pi_2(c(u_d))\ldots\pi_2(c(u_0))=\omega_\ell\ldots\omega_{\ell+d}$ for some $\ell$. Note that since $\omega$ is a de Bruijn word, this $\ell$ is unique. Let us prove by induction on $\ell \in \{1, \ldots, t-d\}$ that $d(u_0,S) \leqslant \ell+d$.
		\begin{itemize}
			\item For $\ell=1$, we have $\pi_2(c(u_d))\ldots\pi_2(c(u_0)) = \omega_1\ldots\omega_{d+1}$. Let us prove that \mbox{$d(u_0,S) \leqslant d+1$}. It is sufficient to prove that $d(u_1,S) \leqslant d$. If $u_1$ accepts at step~(i), the conclusion holds. So we can assume that $u_1$ accepts at step~(v). By~(ii) and~(iii), there exists a decreasing path $u_1v_1\ldots v_d$ (starting at~$u_1$) such that $\pi_2(c(v_d))\ldots\pi_2(c(v_1))\pi_2(c(u_1))$ is a factor of~$\omega$. Since $u_1$ does not reject at step~(iv), $\pi_2(c(v_d))\ldots\pi_2(c(v_1))\pi_2(c(u_1))\pi_2(c(u_0))$ is a factor of $\omega$. Note that $u_0u_1v_1\ldots v_{d-1}$ is a decreasing path. Since $u_0$ does not reject at step~(iii), we get $\pi_2(c(v_{d-1}))\ldots\pi_2(c(v_1))\pi_2(c(u_1))\pi_2(c(u_0)) = \pi_2(c(u_d))\ldots\pi_2(c(u_0))$. Thus, the factor $\omega_1\ldots\omega_{\ell+1}=\pi_2(c(u_d))\ldots\pi_2(c(u_0))$ appears at least twice in $\omega$, which is a contradiction. Hence, $u_1$ accepts at step (i), so $d(u_1,S) \leqslant d$.
			\item Assume now that $\ell \geqslant 2$. To prove that $d(u_0,S) \leqslant \ell+d$, it is sufficient to prove that $d(u_1,S) \leqslant \ell+d-1$. If $u_1$ accepts at step~(i) the conclusion holds. So we can assume that $u_1$ accepts at step~(v). By~(ii) and~(iii), there exists a decreasing path $u_1v_1\ldots v_d$ (starting at $u_1$) such that $\pi_2(c(v_d))\ldots\pi_2(c(v_1))\pi_2(c(u_1))$ is a factor of~$\omega$. Since $u_1$ does not reject at step~(iv), $\pi_2(c(v_d))\ldots\pi_2(c(v_1))\pi_2(c(u_1))\pi_2(c(u_0))$ is a factor of~$\omega$. Note that $u_0u_1v_1\ldots v_{d-1}$ is a decreasing path. Since $u_0$ does not reject at step~(iii), then $\pi_2(c(v_{d-1}))\ldots\pi_2(c(v_1))\pi_2(c(u_1))\pi_2(c(u_0)) = \pi_2(c(u_d))\ldots\pi_2(c(u_0))$.
			Since the factor $\pi_2(c(u_d))\ldots\pi_2(c(u_0))$ appears exactly once in $\omega$ and $\pi_2(c(u_d))\ldots\pi_2(c(u_0))=\omega_\ell\ldots\omega_{\ell+d}$, we have $\pi_2(c(v_d))\ldots\pi_2(c(v_1))\pi_2(c(u_1)) = \omega_{\ell-1}\ldots\omega_{\ell+d-1}$ by~(iv). By induction hypothesis, we get $d(u_1,S) \leqslant \ell+d-1$ and then the conclusion holds.
		\end{itemize}
		Thus, for any vertex $u_0 \in V$, if $u$ accepts at step (i) then $d(u,S) \leqslant d$, otherwise we have $d(u,S) \leqslant t$ by the previous induction. So $S$ is indeed a dominating set at distance $t$.
	\end{proof}
	
	\section{Perfect matchings}
	
	In this section, we will consider the certification of the existence of a \emph{perfect matching}. Let us recall that a perfect matching of a graph $G$ is a set $M \subseteq E(G)$ which satisfies the following property: for all $v \in V(G)$, there exists a unique $e \in M$ such that $v\in e$.

	\subsection{Upper bounds}
	
	In order to certify perfect matching, we will certify matching colorings which are defined as follows:
	
	\begin{definition}
		Let $G=(V,E)$ be a graph. A \emph{$k$-matching coloring} of~$G$ is a mapping \mbox{$\varphi : V \rightarrow \{1, \ldots, k\}$} such that there exists a perfect matching $M$ satisfying the following property: for every edge $(u,v) \in E$, $\varphi(u)=\varphi(v)$ if and only if $(u,v) \in M$.
	\end{definition}
	
	This is not the same definition as the one given in the introduction, but the two are equivalent.
	
	For every graph $G$, we denote by $\Delta(G)$ (or simply $\Delta$ when $G$ is clear from context) the maximum degree of $G$.
	
	\begin{lemma}
		\label{lem:coloring_of_a_matching}
		Let $G=(V,E)$ be a graph with a perfect matching $M$. There exists a $(2\Delta-1)$-matching coloring of $G$.
	\end{lemma}
	
	\begin{proof}
		We construct a $(2\Delta-1)$-matching coloring $\varphi$ of $G$ as follows: we greedily color one after the other the edges of $M$. For an edge $(u,v) \in M$, $u$ has at most $\Delta-1$ neighbors different from $v$, and $v$ has at most $\Delta-1$ neighbors different from $u$, each of them being incident to at most one edge of $M$. Thus, there are at most $2\Delta-2$ forbidden colors for the edge $(u,v)$ in total. So there is at least one available color to color $u$ and $v$.
	\end{proof}
	
	\ThmUBPerfectMatching*

	\begin{proof}[Proof of Theorem~\ref{thm:UB_perfect_matching}]
		Let $G=(V,E)$ be a graph in $\mathcal{C}$ having a perfect matching. Let $\varphi$ be a $k$-matching coloring of $G$. For every $v$, the prover gives the certificate $c(v):=\varphi(v)$ to the vertex $v$.
		
		For every $G\in \mathcal{C}$ and certificate function $c$, every vertex $v$ checks the certificate as follows: $v$ checks if it has exactly one neighbor~$u$ such that $c(u)=c(v)$. If it is the case, $v$ accepts. Otherwise, $v$ rejects.
		
		It remains to show that a graph $G\in \mathcal{C}$ is accepted if and only if it has a perfect matching. If $G$ has a perfect matching, then every vertex accepts with the certificate function given by the prover described above, by definition of a matching coloring. Conversely, if all the vertices accept, we construct a perfect matching by matching each vertex $v$ to the unique vertex $u$ in his neighborhood such that $c(u)=c(v)$.
	\end{proof}
	
	Lemma~\ref{lem:coloring_of_a_matching} and Theorem~\ref{thm:UB_perfect_matching} directly implies that the following holds:

	\CoroMatchingGUB*
	
	We will see in Section~\ref{sec:matchingLB} that the dependency on $\Delta$ is unavoidable in the anonymous model, and unavoidable in the locally checkable proofs model when verification algorithms consist in checking for a matching coloring.
	
	One can wonder if the dependency on $\Delta$ is always unavoidable, or if there exist some non-trivial graph classes with a certification of constant size. We give a very short proof that such graph classes exist, \emph{e.g.} planar graphs: 
	
	\begin{theorem}
		\label{thm:PM_planar}
		Let $k \in \NN$. Let $\mathcal{C}$ be a class of graphs closed by edge contraction such that for all $G\in\mathcal{G}$, the chromatic number of $G$ is at most $k$. Then, $k$ certificates are enough to certify a perfect matching for graphs in $\mathcal{C}$.
	\end{theorem}
	
	\begin{proof}
		By Theorem~\ref{thm:UB_perfect_matching}, it is sufficient to show that for every $G\in\mathcal{C}$, if $G$ has a perfect matching then $G$ has a $k$-matching coloring.
		
		Let $G\in\mathcal{C}$ such that $G$ has a perfect matching $M$. Let us define a $k$-matching coloring~$\varphi$ of $G$. Let $G'$ be the graph obtained from $G$ by contracting every edge of $M$. Since $\mathcal{C}$ is closed by edge contraction, $G'\in\mathcal{C}$ so $G'$ also is $k$-colorable. Let $\varphi'$ be a proper $k$-coloring of $G'$. For every vertex $u\in V(G)$, there exists a unique edge $e\in M$ such that $u\in e$. Let $u'\in V(G')$ be the vertex of $G'$ resulting from the contraction of $e$. Let $\varphi(u):=\varphi'(u')$. Then, $\varphi$ is a $k$-matching coloring of $G$ (because $M$ is a perfect matching and $\varphi'$ is a proper $k$-coloring of~$G'$).
	\end{proof}
	
	The following corollary is a direct consequence of Theorem~\ref{thm:PM_planar} and classical theorems of structural graph theory.
	
	\CoroMinorTreewidth*
	
	\begin{proof}
		The first statement is a consequence of the $4$-color theorem. The second point follows from the best bound known on Hadwiger's conjecture by Delcourt and Postle~\cite{DelcourtP21}. The last statement simply follows from the fact that edge contractions cannot increase the treewidth and graphs of treewidth $k$ are $k$-degenerate (and then $(k+1)$-colorable).
	\end{proof}
	
	Note that the behavior on planar graphs is completely different in the three properties we considered: for coloring, we get a tight bound of $4$ certificates, for dominating sets at distance $t$, the general lower bounds hold, and for certification of perfect matchings, the natural approach which gives $O(\Delta)$ certificates can be drastically improved.
	

	\subsection{Lower bounds}\label{sec:matchingLB}
	
	The first result of this section consists in proving that the number of bits obtained in Corollary~\ref{coro:matching_GUB} to certify that a graph admits a perfect matching is (essentially) tight.
	
	\ThmLBPerfectMatching*
	
	\begin{proof}
		Assume by contradiction that $\Delta-1$ certificates are sufficient to certify the existence of a perfect matching for graphs with maximum degree at most $\Delta$.
		
		Let $B_\Delta = (V_\Delta,E_\Delta)$ be the bipartite graph on vertex set $V_\Delta:=\{u_1, \ldots, u_\Delta\}\cup\{v_1, \ldots, v_\Delta\}$, and $(u_i,v_j)$ is an edge for every $i\geqslant j$. This graph is usually called a \emph{half graph}.
		The graph~$B_\Delta$ has a perfect matching (consisting of the edges $(u_i,v_i)$ for all $i \in \{1, \ldots, \Delta\}$). Thus, there exists a certificate function $c : V_\Delta \rightarrow \{1, \ldots, \Delta-1\}$ which makes every vertex accept. It is easy to check that there actually is a unique perfect matching in $B_\Delta$. By pigeonhole principle, there must exist $1 \leqslant j_1<j_2 \leqslant \Delta$ such that \mbox{$c(v_{j_1})=c(v_{j_2})$}.
		See Figure~\ref{fig:B_7} for illustration.
		
		\begin{figure}[h]
			\centering
			\begin{tikzpicture}
				\node[circle, draw, minimum size=0.8cm, fill=magenta!50] at (0, 3) (1) {\footnotesize$u_1$};
				\node[circle, draw, minimum size=0.8cm, fill=blue!50] at (2.2, 3) (2) {\footnotesize$u_2$};
				\node[circle, draw, minimum size=0.8cm, fill=yellow!50] at (4.4, 3) (3) {\footnotesize$u_3$};
				\node[circle, draw, minimum size=0.8cm, fill=green!50] at (6.6, 3) (4) {\footnotesize$u_4$};
				\node[circle, draw, minimum size=0.8cm, fill=blue!50] at (8.8, 3) (5) {\footnotesize$u_5$};
				\node[circle, draw, minimum size=0.8cm, fill=orange!50] at (11, 3) (6) {\footnotesize$u_6$};
				\node[circle, draw, minimum size=0.8cm, fill=green!50] at (13.2, 3) (7) {\footnotesize$u_7$};
				\node[circle, draw, minimum size=0.8cm, fill=orange!50] at (0, 0) (8) {\footnotesize$v_1$};
				\node[circle, draw, minimum size=0.8cm, fill=yellow!50] at (2.2, 0) (9) {\footnotesize$v_2$};
				\node[circle, draw, minimum size=0.8cm, fill=magenta!50, line width=1.7pt] at (4.4, 0) (10) {\footnotesize$v_3$};
				\node[circle, draw, minimum size=0.8cm, fill=green!50] at (6.6, 0) (11) {\footnotesize$v_4$};
				\node[circle, draw, minimum size=0.8cm, fill=teal!50] at (8.8, 0) (12) {\footnotesize$v_5$};
				\node[circle, draw, minimum size=0.8cm, fill=magenta!50, line width=1.7pt] at (11, 0) (13) {\footnotesize$v_6$};
				\node[circle, draw, minimum size=0.8cm, fill=blue!50] at (13.2, 0) (14) {\footnotesize$v_7$};
				
				\draw
				(1) -- (8)
				(2) -- (8)
				(2) -- (9)
				(3) -- (8)
				(3) -- (9)
				(3) -- (10)
				(4) -- (8)
				(4) -- (9)
				(4) -- (10)
				(4) -- (11)
				(5) -- (8)
				(5) -- (9)
				(5) -- (10)
				(5) -- (11)
				(5) -- (12)
				(6) -- (8)
				(6) -- (9)
				(6) -- (10)
				(6) -- (11)
				(6) -- (12)
				(6) -- (13)
				(7) -- (8)
				(7) -- (9)
				(7) -- (10)
				(7) -- (11)
				(7) -- (12)
				(7) -- (13)
				(7) -- (14);
			\end{tikzpicture}
			\caption{The graph $B_7$ with an accepting certificate assignment (that we will use in the next pictures). Here, certificates are represented by the colors, and only 6 different certificates are used. In this example, we take $j_1=3$ and $j_2=6$.}
			\label{fig:B_7}
		\end{figure}
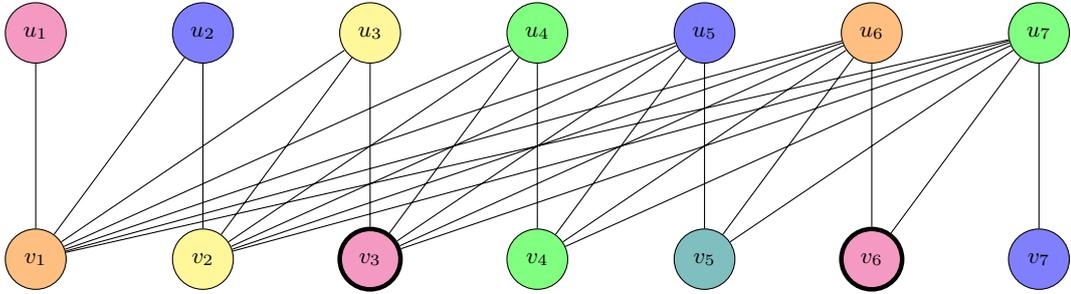
		
		To obtain a contradiction, we construct a graph $G_\Delta$ which does not have a perfect matching, but has a certificate function such that all the vertices accept. We construct $G_\Delta$ as follows. We take two copies of $B_\Delta$, denoted by $B_\Delta'$ and $B_\Delta''$. We denote by $V_\Delta' = \{u_1', v_1', \ldots\}$ (resp. $V_\Delta'' = \{u_1'', v_1'', \ldots\}$) the vertex set of $B_\Delta'$ (resp. $B_\Delta''$). For all $i \in \{j_1, \ldots, j_2-1\}$, we delete the edge $(u_i'',v_{j_1}'')$ and we add the edge $(u_i'',v_{j_2}')$. Finally, the prover gives to the vertices of $B_\Delta'$ and $B_\Delta''$ their certificates in $B_\Delta$. For instance, with the certificates of $B_7$ represented in Figure~\ref{fig:B_7}, the graph $G_7$ receives the certificates represented in Figure~\ref{fig:G_7}.
		
		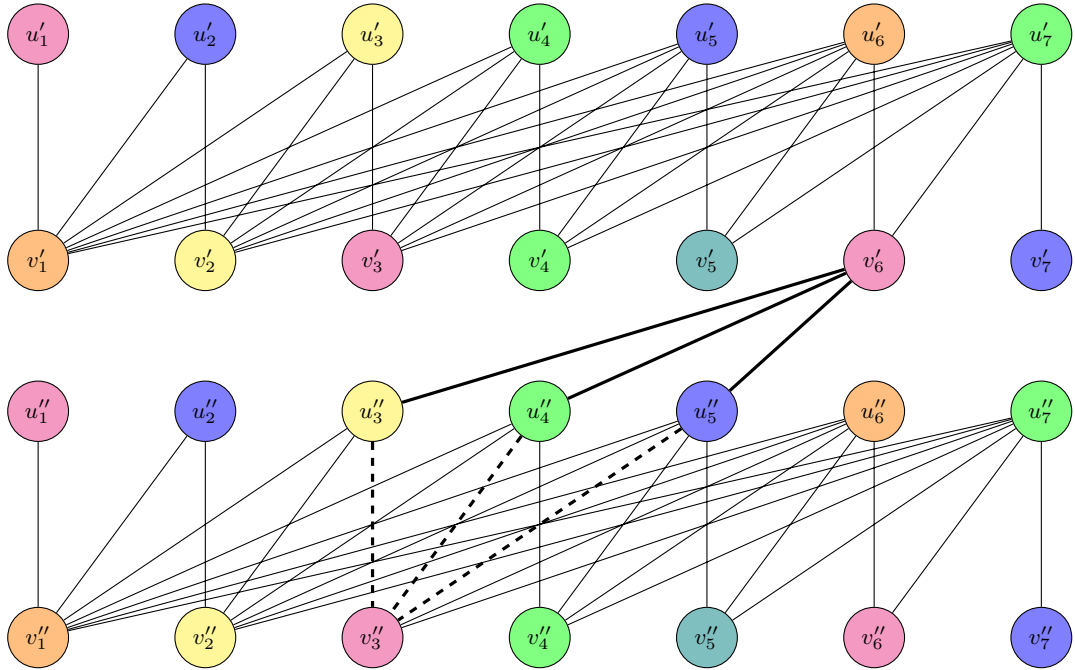
\begin{figure}[h]
			\centering
			\begin{tikzpicture}
				\node[circle, draw, minimum size=0.8cm, fill=magenta!50] at (0, 3) (1) {\footnotesize$u_1'$};
				\node[circle, draw, minimum size=0.8cm, fill=blue!50] at (2.2, 3) (2) {\footnotesize$u_2'$};
				\node[circle, draw, minimum size=0.8cm, fill=yellow!50] at (4.4, 3) (3) {\footnotesize$u_3'$};
				\node[circle, draw, minimum size=0.8cm, fill=green!50] at (6.6, 3) (4) {\footnotesize$u_4'$};
				\node[circle, draw, minimum size=0.8cm, fill=blue!50] at (8.8, 3) (5) {\footnotesize$u_5'$};
				\node[circle, draw, minimum size=0.8cm, fill=orange!50] at (11, 3) (6) {\footnotesize$u_6'$};
				\node[circle, draw, minimum size=0.8cm, fill=green!50] at (13.2, 3) (7) {\footnotesize$u_7'$};
				\node[circle, draw, minimum size=0.8cm, fill=orange!50] at (0, 0) (8) {\footnotesize$v_1'$};
				\node[circle, draw, minimum size=0.8cm, fill=yellow!50] at (2.2, 0) (9) {\footnotesize$v_2'$};
				\node[circle, draw, minimum size=0.8cm, fill=magenta!50] at (4.4, 0) (10) {\footnotesize$v_3'$};
				\node[circle, draw, minimum size=0.8cm, fill=green!50] at (6.6, 0) (11) {\footnotesize$v_4'$};
				\node[circle, draw, minimum size=0.8cm, fill=teal!50] at (8.8, 0) (12) {\footnotesize$v_5'$};
				\node[circle, draw, minimum size=0.8cm, fill=magenta!50] at (11, 0) (13) {\footnotesize$v_6'$};
				\node[circle, draw, minimum size=0.8cm, fill=blue!50] at (13.2, 0) (14) {\footnotesize$v_7'$};
				
				\node[circle, draw, minimum size=0.8cm, fill=magenta!50] at (0, -2) (15) {\footnotesize$u_1''$};
				\node[circle, draw, minimum size=0.8cm, fill=blue!50] at (2.2, -2) (16) {\footnotesize$u_2''$};
				\node[circle, draw, minimum size=0.8cm, fill=yellow!50] at (4.4, -2) (17) {\footnotesize$u_3''$};
				\node[circle, draw, minimum size=0.8cm, fill=green!50] at (6.6, -2) (18) {\footnotesize$u_4''$};
				\node[circle, draw, minimum size=0.8cm, fill=blue!50] at (8.8, -2) (19) {\footnotesize$u_5''$};
				\node[circle, draw, minimum size=0.8cm, fill=orange!50] at (11, -2) (20) {\footnotesize$u_6''$};
				\node[circle, draw, minimum size=0.8cm, fill=green!50] at (13.2, -2) (21) {\footnotesize$u_7''$};
				\node[circle, draw, minimum size=0.8cm, fill=orange!50] at (0, -5) (22) {\footnotesize$v_1''$};
				\node[circle, draw, minimum size=0.8cm, fill=yellow!50] at (2.2, -5) (23) {\footnotesize$v_2''$};
				\node[circle, draw, minimum size=0.8cm, fill=magenta!50] at (4.4, -5) (24) {\footnotesize$v_3''$};
				\node[circle, draw, minimum size=0.8cm, fill=green!50] at (6.6, -5) (25) {\footnotesize$v_4''$};
				\node[circle, draw, minimum size=0.8cm, fill=teal!50] at (8.8, -5) (26) {\footnotesize$v_5''$};
				\node[circle, draw, minimum size=0.8cm, fill=magenta!50] at (11, -5) (27) {\footnotesize$v_6''$};
				\node[circle, draw, minimum size=0.8cm, fill=blue!50] at (13.2, -5) (28) {\footnotesize$v_7''$};
				
				\draw
				(1) -- (8)
				(2) -- (8)
				(2) -- (9)
				(3) -- (8)
				(3) -- (9)
				(3) -- (10)
				(4) -- (8)
				(4) -- (9)
				(4) -- (10)
				(4) -- (11)
				(5) -- (8)
				(5) -- (9)
				(5) -- (10)
				(5) -- (11)
				(5) -- (12)
				(6) -- (8)
				(6) -- (9)
				(6) -- (10)
				(6) -- (11)
				(6) -- (12)
				(6) -- (13)
				(7) -- (8)
				(7) -- (9)
				(7) -- (10)
				(7) -- (11)
				(7) -- (12)
				(7) -- (13)
				(7) -- (14)
				
				(15) -- (22)
				(16) -- (22)
				(16) -- (23)
				(17) -- (22)
				(17) -- (23)
				(18) -- (22)
				(18) -- (23)
				(18) -- (25)
				(19) -- (22)
				(19) -- (23)
				(19) -- (25)
				(19) -- (26)
				(20) -- (22)
				(20) -- (23)
				(20) -- (24)
				(20) -- (25)
				(20) -- (26)
				(20) -- (27)
				(21) -- (22)
				(21) -- (23)
				(21) -- (24)
				(21) -- (25)
				(21) -- (26)
				(21) -- (27)
				(21) -- (28);
				
				\draw
				(13) edge[line width=1.2pt] (17)
				(13) edge[line width=1.2pt] (18)
				(13) edge[line width=1.2pt] (19)
				(17) edge[line width=1.2pt, dashed] (24)
				(18) edge[line width=1.2pt, dashed] (24)
				(19) edge[line width=1.2pt, dashed] (24);;
			\end{tikzpicture}
			\caption{The graph $G_7$, with the certificates used in the proofs. Dashed edges represent edges that have been removed when building $G_7$. The sets used in Claim \protect \textit{}\ref{claim:no_PM_in_Gdelta} are $A=\{v_3'',v_6'',v_7''\}$ and $N(A)=\{u_6'',u_7''\}$.}
			\label{fig:G_7}
		\end{figure}
		
		\begin{claim}
			\label{claim:no_PM_in_Gdelta}
			The graph $G_\Delta$ does not have a perfect matching.
		\end{claim}
		
		\begin{proof}
			The set $A=\{v_{j_1}'', v_{j_2}'', v_{j_2+1}'', \ldots, v_\Delta''\}$ is an independent set of size $\Delta-j_2+2$. And $N(A)$ is $\{u_{j_2}'', \ldots, u_\Delta''\}$, which is a set of size $\Delta-j_2+1$. Therefore, not all vertices of $A$ can be matched. 
		\end{proof}
		
		We show that the certificates given by the prover as described previously make all vertices of $G_\Delta$ accept. In the graph $G_\Delta$, all the vertices of $B_\Delta'$ have the same view as their copy in~$B_\Delta$, except $v_{j_2}'$ which has the view of $v_{j_1}$ in~$B_\Delta$. 
		Since all the vertices of $B_\Delta$ accept, all the vertices of $B_\Delta'$ accept as well. 
		Similarly, in~$B_\Delta''$, all the vertices have the same view as their copy in $B_\Delta$, except $v_{j_1}''$ which has the view of $v_{j_2}$ in $B_\Delta$. Note that, for $u_{j_1}'', \ldots, u_{j_2-1}''$, it is because $c(v_{j_1})=c(v_{j_2})$. Thus, all the vertices of $B_\Delta''$ accept.
		So all the vertices of $G_\Delta$ accept. Together with Claim~\ref{claim:no_PM_in_Gdelta}, this is a contradiction.
	\end{proof}

	We can now wonder if fewer certificates are needed in the locally checkable proofs model. We did not succeed to answer this question in full generality, but we prove the result for a specific type of certification. 
	Remember that in the proof of Theorem~\ref{thm:UB_perfect_matching} with the matching coloring, the vertices check independently their neighbors in order to look for an edge of the matching. Every vertex then accepts if it is adjacent to exactly one edge of the matching. In other words, by looking only at the certificates and identifiers of its endpoints, each edge can be described as \emph{valid} or \emph{invalid}, in the sense that the nodes consider it to be an edge of the matching or not.
	Each vertex accepts if and only if it is the endpoint of exactly one valid edge. We call such a checking algorithm a~\emph{constructive} checking. If we have a constructive checking which makes every vertex accept, it is possible to construct a perfect matching simply by matching each vertex with the endpoint of its valid edge. Note that, as in the proof of Theorem~\ref{thm:UB_perfect_matching}, $2\Delta-1$ different certificates are always sufficient to certify a perfect matching with a constructive checking. In fact, we will show that it is optimal (up to a constant factor on the number of bits).
	
	\begin{theorem}
		\label{thm:LB_constructive_PM}
		In the locally checkable proofs model, for every $\Delta \geqslant 1$, at least $\Omega(\log \Delta)$ bits are needed to certify a perfect matching with a constructive checking.
	\end{theorem}
	
	\begin{proof}
		Assume that $m$ bits are sufficient to certify a perfect matching with a constructive checking.
		Let us fix $2\Delta$ different identifiers, denoted by $Id_1, \ldots, Id_{2\Delta}$. Let us consider the bipartite graph $B_\Delta = (V_\Delta, E_\Delta)$ introduced in the proof of Theorem~\ref{thm:LB_perfect_matching}. For each permutation $\sigma$ of $\{1, \ldots, \Delta\}$, let us consider $B_\Delta(\sigma)$ the graph where each vertex $u_i$ has the identifier $Id_{\sigma(i)}$, and each vertex $v_i$ the identifier $Id_{\Delta+i}$. Since $B_\Delta(\sigma)$ has a perfect matching, there exists a certificate function \mbox{$c_\sigma : \{Id_1, \ldots, Id_{2\Delta}\} \rightarrow \{0, \ldots, 2^m-1\}$} which makes all vertices of $B_\Delta(\sigma)$ accept.
		
		\begin{claim}
			\label{claim:constructive_verification_PM}
			Let $\sigma, \tau$ be two permutations of $\{1, \ldots, \Delta\}$ such that $c_\sigma = c_\tau$. Then, $\sigma=\tau$.
		\end{claim}
		\begin{proof}[Proof of Claim~\ref{claim:constructive_verification_PM}]
			Let $\sigma$, $\tau$ be such that $c_\sigma=c_\tau$. By definition of a constructive verification, the set of valid edges in the graph $B_\Delta(\sigma)$ certified with $c_\sigma$ is a perfect matching. Since the only perfect matching in $B_\Delta$ is $\{(u_1,v_1), \ldots, (u_\Delta, v_\Delta)\}$, it is the set of valid edges in $B_\Delta(\sigma)$ certified with $c_\sigma$.
			
			By contradiction, assume that $\sigma\neq\tau$. Let $i_0$ the smallest element of $\{1, \ldots, \Delta\}$ such that $\sigma(i_0)\neq\tau(i_0)$. In $B_\Delta(\tau)$, the edge $(u_{i_0}, v_{i_0})$ is valid (because valid edges form a perfect matching). Let $j_0$ be such that vertex $\tau(j_0)=\sigma(i_0)$. By definition of $i_0$, we get $j_0 \in \{i_0+1, \ldots, \Delta\}$. Since $c_\sigma=c_\tau$, the edge $(u_{j_0}, v_{i_0})$ has same identifiers and certificates than edge $(u_{i_0}, v_{i_0})$ in $B_\Delta(\sigma)$. Thus, the edge $(u_{j_0}, v_{i_0})$ is also valid in $B_\Delta(\tau)$. This is a contradiction because vertex $v_{i_0}$ would then be the endpoint of two valid edges in $B_\Delta(\tau)$, so it should reject.
		\end{proof}
		We can now conclude the proof of Theorem~\ref{thm:LB_constructive_PM}. Assume that $m$ bits are sufficient to certify a perfect matching in the locally checkable proofs model, with a constructive verification algorithm. By Claim~\ref{claim:constructive_verification_PM}, the number of different functions $\{Id_1, \ldots, Id_{2\Delta}\} \rightarrow \{0, \ldots, 2^m-1\}$ is at least the number of permutations of $\{1, \ldots, \Delta\}$. Thus, we get \mbox{$2^{2m\Delta} \geqslant \Delta!$}, leading to $m \geqslant \frac{\log_2(\Delta!)}{2\Delta}$. Since $\log_2(\Delta!)=\Omega(\Delta\log\Delta)$, we get the result.
	\end{proof}

	\bibliographystyle{plain}
	\bibliography{biblio}

\end{document}

%% file: intro.tex
\section{Introduction}
\label{sec:introduction}

\textbf{Local certification.} 
Local certification is a topic at the intersection of locality and fault-tolerance in distributed computing. 
Very roughly, the main concern is to measure how much information the nodes of a network need to know in order to verify that the network satisfies a given property. 
Local certification originates from self-stabilization and is tightly related to the minimal memory needed to be sure that a self-stabilizing algorithm has reached a correct configuration locally. 
The topic is now studied independently, and the area has been very active during the last decade. 
We refer to the survey~\cite{Feuilloley21} for an introduction to the topic.

The standard model for local certification is the following. 
The nodes are first assigned labels, called \emph{certificates}, and then every node looks at its certificate and the certificates of its neighbors, and decides to accept or reject. 
A certification scheme for a given property is correct if, for any network, the property is satisfied if and only if there exists a certificate assignment such that all the nodes accept. (We discuss  variations later, and give proper definitions in Section~\ref{sec:model}.) 
The usual measure of performance of a certification scheme is the maximum certificate size over all nodes, and all networks of a given size. 

One of the fundamental results in local certification is that any property can be certified if the network is equipped with unique identifiers~\cite{KormanKP10, GoosS16},
but this is at the expense of huge certificates. 
Indeed, the scheme consists in giving to every node the full map of the graph, which takes $\Theta(n^2)$ bits in $n$-node graphs.
The question then is: when can we do better? 
There exist three typical certificate sizes. 
For some properties, \emph{e.g.} related to graph isomorphism, $\Theta(n^2)$ bits is the best we can do~\cite{GoosS16}.
For many natural properties, the optimal certificate size is $\Theta(\log n)$; for example most properties related to trees (acyclicity, spanning tree, BFS, and minimum spanning tree for small edge weights~\cite{KormanKP10}). 
A recent research direction tries to capture precisely which properties have such \emph{compact certification} (see~\cite{BousquetFP22, FraigniaudMRT22, FraigniaudMMRT23}). 
Finally, some properties are \emph{local} from the certification point of view, in the sense that the optimal certification size \emph{does not depend on $n$}.
This third type of property is the topic of this paper.%
\medskip

\noindent\textbf{Local certification of local properties.} 
Until recently, studying local properties has not been the focus of the community, since the usual parameter for measuring complexity is the network size.
A recent paper by Ardévol Mart{\'{\i}}nez, Caoduro, Feuilloley, Narboni, Pournajafi and Raymond~\cite{ArdevolCFNPR22} is the first to target this regime. 
We refer to~\cite{ArdevolCFNPR22} for the full list of motivations to study this topic, and we just highlight a few points here. 

First, it will appear in this paper that the size of the certificates for local properties is often expressed as functions of parameters different from the number of nodes. Therefore, one should not (always) see these are constants.
It has been highlighted before (see discussion before Open problem 4 in \cite{Feuilloley21}) that we have basically no understanding of certification size expressed by other parameters than the number of nodes.

Second, \emph{locally checkable languages} (LCL) are at the core of the study of the LOCAL and CONGEST models. 
These are basically the properties that are local from a certification point of view: the output can be checked by looking at all the balls of some constant radius. 
Certification is a way to question the encoding of LCLs. 
A typical example is coloring, for which one uses the colors as output of a construction algorithm, and as certificates for colorability certification. 
If there would exist a better certification, this would shed a new light on the encoding of this LCL.
For example, the celebrated round elimination technique~\cite{Suomela20} is very sensitive to the problem encoding, and one could hope that feeding it with a different encoding could provide new bounds.
In a more general perspective, we argue that just like understanding the complexity of LCLs, understanding certification of local properties is a fundamental topic.

In addition to these two general motivations, our work is guided by two open problems, the $k$-colorability question and the trade-off conjecture, that we detail now.
\medskip

\noindent\textbf{The $k$-colorability question.}
Let us first discuss the case of colorability, which will be central in this paper.
The property we want to certify is that the graph is \emph{$k$-colorable}, that is, one can assign colors from $\{1,\ldots,k\}$ to vertices such that no two neighbors have the same color. 
It is straightforward to design a local certification with $k$ certificates for this property: the certificates encode colors in $\{1,\ldots,k\}$ and the nodes just have to check that there is no conflict.\footnote{A \emph{conflict} being an edge whose two endpoints are colored the same.} This uses $O(\log k)$ bits. 
The natural open problem here is the following.

\begin{open}[Open problem 1 in \cite{Feuilloley21}]\label{open:k-colorability}
	Is $\Theta(\log k)$ optimal for $k$-colorability certification?
\end{open}

The first result on that question is the very recent paper~\cite{ArdevolCFNPR22}  which establishes that one bit is not enough to certify $k$-colorability. 
This lower bound holds in the anonymous and in the proof-labeling scheme models, that we will define later. 
In a nutshell, the technique is an indistinguishability argument: assuming that there exists a 1-bit certification, one can argue about the number of 1s in a node neighborhood and take an accepting certification of some $k$-colorable graph to derive an accepting certification of a $k+1$-clique, which is a contradiction. 

Interestingly,~\cite{ArdevolCFNPR22} also shows a case where the natural encoding is not the best certification. 
Namely, certifying a distance-$2$ $3$-coloring can be done with only $1$ bit, while the obvious encoding of the colors takes $2$ bits.
Finally, let us mention that non-$k$-colorability, the complement property, is much harder to certify. 
Indeed, it is proved in~\cite{GoosS16} that one needs $\tilde{\Omega}(n^2)$ bits to certify non-$3$-colorability (where $\tilde{\Omega}$ hides inverse logarithmic factors).
\medskip

\noindent\textbf{Trade-off conjecture.}
We finish this general introduction, with yet another motivation to study local properties. 
The \emph{trade-off conjecture}, first stated in \cite{FeuilloleyFHPP21}, basically states that for any property, if the optimal certification size is $s$ for the classic certification mechanism, where the nodes see their neighbors certificates, then it is in $O(s/d)$ if the vertices are allowed to see their whole neighborhood at distance $d$. 
This was proved to be true for many classic properties, and in many large graph classes~\cite{FeuilloleyFHPP21, FischerOS21, OstrovskyPR17}.
Implicitly, the big-O of the conjecture refers to functions of $n$, but for local properties, the conjecture is interesting only if it refers to the parameters that appear in the certificate size.

\begin{open}[Trade-off conjecture]
	Consider a property with optimal certification size~$s$ at distance 1 (where $s$ depends on the natural parameters of the problem). Is it true that if we allow the verification algorithm to look at distance $d$, then the optimal size is at most $\alpha\cdot s/d$ for some constant $\alpha$?
\end{open}

The authors of \cite{ArdevolCFNPR22} argue that the conjecture might actually be wrong for local properties.
In other words, there might exist properties such that looking further in the graph is useless (unless you can see the whole graph), or at least not as useful as claimed in \cite{FeuilloleyFHPP21} (instead of being $s/d$ the optimal size could be a less-decreasing function of $d$).
We consider this question to be very intriguing and important to the study of locality, and we will discuss it several times in the paper.
\medskip

\noindent\textbf{A sample of local properties.}
Local properties have different behaviors, which prevented us to establish general theorems capturing all of them. Instead, we looked for a sample of properties widely studied in the distributed community having diverse behaviors, and such that many other properties would behave similarly to one of the sample. 
First, we chose to study the colorability question, for the reasons cited above. 
Second we looked at domination at distance $t$, where a set of nodes is selected, and we want to check that every node is at distance at most $t$ from a selected node. This property has inputs and an external parameter, which makes it very different from colorability. Moreover, a dominating set distance $t$ is a building block for many self-stabilizing algorithms. 
Third, we study the property of ``having a perfect matching''. This differs from the two first ones by being an edge-related problem instead of the node-related problem, and it appears that the key parameter there is the maximum degree, a new parameter for certification. One more motivation is that matchings are classic objects in distributed graph algorithms.
\medskip

\noindent\textbf{Organization of the paper.}
The paper is organized the following way. 
After this introduction, we give a detailed overview of the context, results. We also give the proof techniques of the main results. 
Then, after a definition section, there are three technical sections that correspond to the three local properties we study. 
The overview and the technical parts can be read independently. 
Readers interested in motivations, general discussions, proof ideas and comparison with previous techniques can read the first and cherry-pick specific proofs they are curious about in the second; while others will prefer to go directly to the model section and formal proofs. 
The overview and the technical part are organized in the same way to allow easy back-and-forth reading.

\section{Overview of our results and techniques}

\paragraph*{Quick description of the models}

In order to state the results, we need to informally define the different models of certification (see Section~\ref{sec:model} for more formal definitions).
\begin{itemize}
	\item In the anonymous model, the nodes have no identifiers/port-numbers.
	\item In the proof-labeling scheme model, every node has a unique identifier (encoded on $O(\log n)$ bits) but it cannot access the identifiers of the other nodes.
	\item In the locally checkable proof model, nodes have identifiers and can see the identifiers of the other nodes.
\end{itemize}

The anonymous model is the one for which we have the largest number of results. It is usually less considered in the literature, but argue that for local properties it is the most natural. See the discussion in Subsection~\ref{subsec:model-certification}.

The standard assumption is that the nodes can only see their neighbors. 
Since we are interested in the trade-off conjecture, we will also consider certification at distance $d$, where the view is the full neighborhood at distance $d$. 

It is often handy to say that the certificates are given by a \emph{prover}, that intuitively tries to convince the nodes that the network is correct (both on correct and incorrect instances).

Now that we are equipped with these notions, we will review our results and techniques, problem by problem.

\subsection{Overview for colorability}

In this subsection, we consider the \emph{$k$-colorability} property, already mentioned, which states that the graph is $k$-colorable. 

\paragraph*{Two lower bounds for colorability}

We have already discussed the colorability property, and cited Open problem~\ref{open:k-colorability}.
Our first result is in the anonymous model, where we strongly solve Open problem~\ref{open:k-colorability} by determining
the \emph{exact bound}: in the anonymous model it is necessary and sufficient to have $k$ different certificates.

\begin{restatable}{theorem}{ThmColAnon}
	\label{thm:col_anon}
	For every $k \ge 2$, in the anonymous model where vertices can see at distance~1, $k$~certificates are needed in order to certify that a graph is $k$-colorable. 
	Therefore, in the anonymous model, exactly $\lceil \log(k) \rceil$ bits are needed to certify $k$-colorability.
\end{restatable}

The upper bound is trivial since we can simply give the colors as certificates. 
The technique to establish this lower bound is a form of crossing technique. 
We take a large enough complete $k$-partite graph, thus a graph that is maximally $k$-colorable, in the sense that any edge added to it would make it non-$k$-colorable. 
The idea of the proof is to argue by counting, that for any certificate assignment that would make all nodes accept this graph, there must exist two edges that we can cross (that is replace $(a,b);(x,y)$ by $(a,x);(b,y)$ for example) such that all neighborhoods appearing in this instance did appear in the previous one, thus no node rejects. 
In addition, the graph obtained after this crossing is non-$k$-colorable, which is a contradiction with the correctness of the scheme.

Now in the most general model, we also give an (asymptotic) tight bound, fully answering Open question~\ref{open:k-colorability}. We actually prove a more general lower bound parametrized by the verification distance.

\begin{restatable}{theorem}{ThmLBColoringDistanced}
	\label{thm:LB_coloring_distance_d}
	In the locally checkable proofs model, at least  $\Omega(\log(k)/d)$ bits are needed to certify $k$-colorability when the vertices are allowed to see their neighborhoods at distance $d$.
\end{restatable}

We describe the proof of this result in a communication complexity framework to provide more intuition, although the actual proof does not rely on any black-box result from communication complexity.
The instances we use have the typical shape of communication complexity constructions: the graph has two parts that correspond to the players (left and right) and a part in the middle.
The middle part has a large diameter (to be sure that left and right cannot communicate at distance $d$), and has two sets of $k$ special vertices on the left (top left and bottom left), and two sets of $k$ special vertices on the right (top right and bottom right).
The part of the left (resp. right) player is an antimatching between the top left and bottom left (resp. top right and bottom right), where an antimatching is a complete bipartite graph in which a matching has been removed.
The middle part has a very constrained structure whose role is to enforce that the graph is $k$-colorable, if and only if, the left and right antimatchings are intuitively mirrors one of the other. 
The idea is then that the information about the exact matchings on the left and right parts has to be transferred to some node, to be compared, and this can happen only via the certificates. 
There are $k!$ possible forms for the left and right antimatchings, (because there are $k!$ possible matching in a complete bipartite graph of size $k$ basically). 
Thus, the information that has to be transferred from one side of the graph to the other has size $k\log k$ (via Stirling equivalent) and since our graph has cuts of order $k$, we get the $\Omega(\log k)$ lower bound.

\paragraph*{Uniquely $k$-colorable graphs and other natural counterexample candidates}

Our two lower bound constructions have in common to be very constrained, in a precise sense: for every vertex $v$, every ball centered in $v$ of radius at least $2$ admits a unique proper $k$-coloring (up to color renaming). 
In this case, we say that the graph is \emph{uniquely $k$-colorable at distance $d$} (where $d$ is the radius of the neighborhood). It is easy to see that if a graph is uniquely $k$-colorable at distance $d$, then either it has a unique $k$-coloring or it is not $k$-colorable.
Intuitively, graphs with this property are hard for certification, since there is no slack in the coloring, thus the transfer of information between different parts of the graph cannot a priori be compressed. 
Perhaps surprisingly, we prove that for these graphs the trade-off conjecture does hold, even in the anonymous model.

\begin{restatable}{theorem}{ThmUniquelyColorable}
	\label{thm:uniquely_colorable}
	For every $d \leqslant \log k$, in the anonymous model where vertices can see at distance~$d$, $O(\log k/d)$ bits are sufficient to certify that a uniquely $k$-colorable graph at distance $d-2$ is $k$-colorable.
\end{restatable}

Let us briefly explain the main ingredient of that proof. Since the coloring is locally unique, by looking far enough, a node can decide which other nodes are in the same color class in a $k$-coloring, if one exists globally. 
This is not enough to be sure that the graph is $k$-colorable, because the color classes might not coincide nicely. 
For example, every cycle is uniquely 2-colorable at distance 1, but this does not certify that the full cycle is 2-colorable.
The key idea is that a node will recover the name of its  color by gathering the bits of information spread on the nodes \emph{of its own color class} at distance at most $d$ from it. 
This allows to spread the information of the color classes on several vertices and then use smaller certificates. Slightly more formally, the certifier will assign to a well-chosen subset of nodes $X$ a special certificate. The vertices of $X$ are chosen far enough from each other so that we can store information of the different color classes on the nodes close to it. But they are also chosen not too far away from each other to be sure that all the vertices are close to a vertex of $X$. Now every node just have to perform the following verification: (i) if it is too far from any vertex of $X$, it rejects, (ii) it checks that its color is the same for all vertices of $X$ close to it, (iii) it checks that for all the vertices of $X$ close to it, the color associated to it is different from colors given to its neighbors.

After Theorem~\ref{thm:uniquely_colorable}, a question is whether we can go down to a constant number of bits, or in other words, how large the constant of the  big-O needs to be. 
We prove that at a distance $O(\log k)$, \emph{one bit is sufficient} (in other words only two different certificates are needed). Note that it means in particular that we can avoid the special certificate of the proof of Theorem~\ref{thm:uniquely_colorable} by being very careful on how we represent each vertex of $X$ and the colors classes around it (and proving that even if a node makes a mistake in the choice of $X$ when some vertices are indistinguishable, its decision will be anyway correct).

\begin{restatable}{theorem}{ThmTwoCertificates}
	\label{thm:2_certificates_are_sufficient}
	In the anonymous model where vertices can see to distance $d$, $2$ certificates are enough to certify that a uniquely $k$-colorable graph at distance $d-2$ is $k$-colorable, if $d\geqslant 9\lceil \log_2 k\rceil +8$.
\end{restatable}

The scheme for this is built on the same framework, but with more ideas and technicalities, and we refer to the technical part for the details.

Now that we proved that locally uniquely colorable graphs are not going to help, we wonder what would be a good candidate to disprove the trade-off conjecture (of course one might not exist, if the conjecture is true). 
We believe that for graphs that are \emph{almost} locally uniquely colorable (that is, they have few correct colorings locally, up to color renaming) the proof of Theorem~\ref{thm:uniquely_colorable} could be adapted, with more layers of technicalities.
Hence, one might go for graphs that have \emph{many} possible colorings. 
This could make sense from a communication complexity point of view, because there also exist difficult problems with many correct pairs of inputs (\emph{e.g.} the disjointness problem). 
What is problematic here is that unlike in communication complexity, in our case, on a $k$-colorable graph, the prover has the choice of the coloring, thus can choose one that is easier to encode compactly.

Since graphs that are very structured seem to admit a linear scaling, another approach could be to consider graphs with a more chaotic structure.
Random graphs are natural candidates here, but even if we could come up with a satisfying definition for what it means to certify the colorability of a random graph, it is not clear that this would be hard. 
Indeed, there exist efficient algorithms to $k$-color $k$-colorable graphs (\emph{e.g.} \cite{Turner88}) that exploit only the local structure of the graph, so one could even hope for a verification without certificates. 

In the end, graphs with large girth and large chromatic number might be the right type of graphs to consider because they have both a large number of colorings locally and a rigid structure globally. Several constructions for these have been designed (see \emph{e.g.}~\cite{Erdos59,Morgenstern94,Coja-OghlanEH16}), but they are randomized or complex, which makes their study rather challenging.

\paragraph*{Certification versus output encoding}


More generally, it is interesting to explore the links between (1) certifying that a given structure exists (\emph{e.g.} a coloring here, but also a perfect matching a bit later in the paper) (2) explicitly describing it, and (3) implicitly describing it, that is giving enough information to recover the structure, but not directly (\emph{e.g.} with fewer bits than the natural encoding). 
We refer to \cite{BickKO22} for the related topic of distributed zero-knowledge proofs.

For colorability, the case of perfect graphs is a nice example of the discrepancy that can exist between these. 
A graph is \emph{perfect} if its chromatic number is equal to the size of its maximum clique for every induced subgraph. Many classic graph classes are perfect, \emph{e.g.} bipartite graphs, chordal graphs, comparability graphs (see \cite{Trotignon13} for a survey on perfect graphs). 
If we are promised to be in such a graph, the verification (at distance 2) is very easy: a vertex just has to check whether it belongs to a clique of size larger than $k$ (in which case it rejects, otherwise it accepts). 
Thus, no certificates are needed for colorability certification, while it seems really difficult to have the nodes output a coloring without giving them quite a lot of information. Proving formally such a discrepancy would be a nice result.

Finally, let us mention yet another lower bound approach, related to problem encodings. 
One can note that if we have certification using $s$ bits for a property, where the local verification algorithm runs in polynomial time in the network size, then we have a \emph{centralized} decision algorithm for this property in time $s^n\text{poly}(n)$. 
Indeed, one can just enumerate all the possible certificate assignments of this size and check whether one is accepted. 
One could hope that having a too-good bound for certification would imply that some big conjecture (\emph{e.g.} SETH, the strong exponential hypothesis) is wrong, and get a conditional lower bound.\footnote{Such bounds are not common in certification, but not unseen, see~\cite{EmekG20}.}
Unfortunately, this does not help us much for certification, since there are known algorithms for computing the chromatic number of a graph in time $O(c^n)$ with $c<3$ (\emph{e.g.} in~\cite{Byskov04}).

\subsection{Overview for domination}

The \emph{domination at distance $t$} property applies to graphs with inputs: every node should be labeled with $0$ or $1$, and every node should be at distance at most $t$ from a node labeled $1$. 
To avoid confusion, let us highlight that in the domination property, the inputs are part of the instance and are different from the certificates. 

This problem is quite different from colorability, in several ways. 
First, it is a problem that is centered on inputs: for any graph there are inputs that are correct, so in some sense it is the inputs that are certified more than the graph itself. 
(This is actually closer to the original self-stabilizing motivation of certifying the output of an algorithm.)
Second, the natural certification has a different flavor: we give every node its distance to the closest node of the dominating set (that is, a node labeled $1$), and every node checks that these distances do make sense.
One can then consider domination at distance $t$ to be the local analogue of acyclicity, which is certified by providing the nodes of the network with the distance to the root.

Given this analogy with acyclicity, a $\Theta(\log t/d)$ optimal certificate size is expected, and indeed we prove it. 
But we go one step further by providing precise bounds on the number of different certificates.
First, for the anonymous model at distance 1,  we prove that the optimal number of different certificates is basically $\sqrt{t}$.

\begin{restatable}{theorem}{ThmSqrtNeeded}
	\label{thm:sqrt_needed}
	Let $t\in \NN^\ast$. In the anonymous model where vertices can see to distance $1$, at least $\sqrt{t-1}$ different certificates are needed to certify a dominating set at distance $t$, even if the graphs considered are just paths and cycles.
\end{restatable}

The lower bound is again quite expected: the proof for this kind of bound is based on arguing whether the same pair of certificates can appear several times on a path or not, thus the square root pops up naturally. 

\begin{restatable}{theorem}{ThmSqrtSufficient}
	\label{thm:sqrt_sufficient}
	In the anonymous model where vertices can see to distance 1, $3 \cdot \lceil \sqrt{t} \rceil$ certificates are sufficient to certify a dominating set at distance $t$.
\end{restatable}

This upper bound is more surprising. 
It reveals that the natural encoding (consisting of giving the minimum distance to a labeled vertex) is not the best, and that the square root is not an artifact of the lower bound proof but is necessary.
The proof of Theorem~\ref{thm:sqrt_sufficient} is based on an elegant argument using \emph{de Bruijn words}. 
Roughly, for some parameters $k$ and $n$, a de Bruijn word is a (cyclic) word on an alphabet of size $k$, which contains all the factors (that is subwords of consecutive letters) of size $n$ exactly once. 
The idea here is that instead of giving the certificate $r$ to nodes at distance $r$ from a node labeled $1$, we will give them the $r$-th letter in a predefined de Bruijn word that corresponds to $n=2$ and $k=\sqrt{t}$. 
Since every node will see its neighbors' certificates, thus a factor of size at least $2$, it will be able to decode what is its position in the word, thus its distance to the 1-labeled node. 
The parametrization $k=\sqrt{t}$ ensures that the de Bruijn word of the correct length exists. 
Finally, we generalize these bounds to larger distances using generalizations of the techniques described above.

\begin{restatable}{theorem}{ThmLBDominating}
	\label{thm:LB_dominating}
	In the anonymous model where vertices can see to distance $d < \frac{t}{2}$, at least $\sqrt[2d]{t-2d+1}$ different certificates are needed to certify a dominating set at distance $t$, even if the graphs considered are just paths and cycles.
\end{restatable}

\begin{restatable}{theorem}{ThmUBDominating}
	\label{thm:UB_dominating}
	In the anonymous model where vertices are allowed to see to distance $d$, $O(\sqrt[d+1]{t})$ certificates are sufficient to certify a dominating set at distance $t$.
\end{restatable}

\subsection{Overview for perfect matchings}

A graph has the \emph{perfect matching property} if it has a perfect matching, that is, a set of edges such that every vertex belongs to exactly one such edge.
It is yet another type of property, that differs from the others by the fact that it has no built-in parameter (no number $k$ of colors, or distance $t$). 
As we will see, the relevant parameter here is the maximum degree of the graph.

\paragraph*{Perfect matching certification upper bounds}

The natural way to locally encode a matching in distributed computing is to make every node know which port number corresponds to a matched edge (if the vertex is matched). 
In the port-number model, this directly leads to a certification: give the relevant port number to each node, and let them check the consistency. 
This takes $O(\log \Delta)$ bits per node, but it requires port numbers, and the ability for the nodes to know the port numbers of their neighbors. 

Our first result is that we do not need port-numbers (nor any kind of initial symmetry-breaking). 
The strategy for the prover is the following. First, choose a perfect matching, and color the matched edges such that there are no two edges $(u,v)$ and $(w,z)$ of the same color with $(v,w)$ being an edge of the graph. 
We call this a \emph{matching coloring}.
Then, give to every node the color of its matched edge. 
Then every node simply checks that it has exactly one neighbor with the same color.

\begin{restatable}{theorem}{ThmUBPerfectMatching}
	\label{thm:UB_perfect_matching}
	Let $k\in\NN^\ast$. Let $\mathcal{C}$ be a class of graphs such that, for every $G\in\mathcal{C}$, if $G$ has a perfect matching then $G$ has a $k$-matching coloring.
	Then, in the anonymous model at distance 1, $k$ certificates are enough to certify the existence of a perfect matching in $G$ for every $G \in \mathcal{C}$.
\end{restatable}

We can easily show (Lemma~\ref{lem:coloring_of_a_matching}) that if $G$ has a perfect matching then it admits a $(2\Delta-1)$-matching coloring.

\begin{restatable}{corollary}{CoroMatchingGUB}
	\label{coro:matching_GUB}
	For every graph $G$, $2\Delta-1$ certificates are enough to certify the existence of a perfect matching, in the anonymous model at distance 1.
\end{restatable}

To better appreciate the lower bound of the next paragraph, let us mention a surprising result on the upper bound side. 
Note first that if we can get a matching coloring using fewer colors, then we automatically get a more compact certification. 
Now consider a planar graph with an arbitrary perfect matching.
We claim that this perfect matching can be colored with only four colors, even though planar graphs can have an arbitrary large maximum degree. 
Indeed, if we contract the edges of the matching, we still have a planar graph, and we can color this graph with four colors, by the four color theorem. 
Now undoing the contraction, and giving to the matched edges the color of their contracted vertices, we get a proper matching coloring with only four colors.
We also prove a more general result on minor-free and bounded treewidth graphs. 

\begin{restatable}{corollary}{CoroMinorTreewidth}
	\label{coro:minor-treewidth}
	In the anonymous model at distance 1:
	\begin{itemize}
		\item Only $2$ bits are enough to certify the existence of a perfect matching for planar graphs.
		\item Only $O(\log k)$ bits are needed to certify the existence of a perfect matching in $K_k$-minor-free graphs.
		\item Only $\lceil\log_2(k+1)\rceil$ bits are needed to certify the existence of a perfect matching in graphs of treewidth at most $k$.
	\end{itemize}
\end{restatable}

Note that, all our upper bounds follow from the existence of $k$-matching colorings. One can easily remark that, if instead of certifying a perfect matching one is simply interested in representing and certifying a matching, we can also do it with the same method by simply coloring all the unmatched vertices with an additional color. Since all the graph classes we mention in the upper bounds are closed under vertex deletion, our results ensure that the following holds: We can represent and certify matchings with $2 \Delta$ certificates for general graphs, $5$ certificates for planar graphs and $t+2$ certificates for graphs of treewidth at most~$t$.

\paragraph*{Perfect matching certification lower bounds}

We prove the following lower bound for perfect matching certification.

\begin{restatable}{theorem}{ThmLBPerfectMatching}
	\label{thm:LB_perfect_matching}
	For every $\Delta \geqslant 2$, in the anonymous model where vertices can see at distance~1, $\Delta$ different certificates are needed to certify the existence of a perfect matching for graphs of maximum degree $\Delta$. 
\end{restatable}

Before we sketch the proof, let us note that there is no natural candidate for a lower bound: on the one hand, graphs that are sparse are ruled out by Corollary~\ref{coro:minor-treewidth}, and on the other hand many dense graph classes are known to always have perfect matchings, \emph{e.g.} even cliques, or even random graphs with at least $n\log n$ edges~\cite{ErdosR66}.

The key to the proof is the notion of \emph{half-graphs}. 
Half-graphs are bipartite graphs with vertices $u_1,...,u_n$ and $v_1,...,v_n$  such that $u_i$ is linked to $v_1,...,v_i$. 
In such a graph, there exists a \emph{unique} perfect matching, and it uses the edges $(u_i,v_i)$. 
Indeed, $u_1$ must be matched with $v_1$, thus $u_2$ must be matched with $v_2$ etc. 
Now, we prove by counting argument that if such a graph is accepted with fewer than $\Delta$ certificates, we can create a graph without perfect matching that is also accepted. 
The idea is to take two copies of this certified graph, and then to carefully remove edges from within these graphs to add edges in between. 
Again, we refer to the technical section for the details.

\paragraph*{Discussion of the parameter $\Delta$}

A remarkable feature of the perfect matching property is that the optimal certificate size is completely captured by the maximum degree (if we do not consider restricted graph classes). 
As far as we know, it is the first time that $\Delta$ appears as a natural parameter for local certification. 
This is interesting, since there are few results that use graph parameters other than the size of the graph to measure certificate size. 
To our knowledge, the only pure graph parameter that has been used before and that does not appear as a parameter of the problem is the girth, for approximate certification~\cite{EmekG20}.\footnote{The maximum edge weight also appears for problems in weighted graphs, \emph{e.g.} minimum spanning tree~\cite{KormanKP10}, max-weight matching~\cite{GoosS16, Censor-HillelPP20}.}
It has been highlighted, \emph{e.g.} in~\cite{Feuilloley21} (discussion before open problem 5), that developing a theory of parametrized certification is an interesting research direction.

Let us note however that it is maybe not very surprising that the maximum degree appears as a key parameter for matching-related problems, since celebrated papers have proved that it is central to the complexity of such problems in other models of computation, \emph{e.g.} the LOCAL model~\cite{BalliuBHORS21,BrandtO20}.